%% file: main.tex
\documentclass[12pt]{article}

\usepackage[letterpaper,top=1in,bottom=1in,left=1in,right=1in]{geometry}

\usepackage{upgreek}

\usepackage{titlesec}
\titleformat{\section}{\normalfont\scshape\large}{\S\thesection.}{0.5em}{}
\titleformat{\subsection}{\normalfont\scshape\normalsize}{\thesubsection.}{0.5em}{}
\titleformat{\subsubsection}{\normalfont\itshape\normalsize}{\thesubsubsection.}{0.5em}{}

\usepackage{amsthm}
\newtheorem{theorem}{Theorem}[section]
\newtheorem{lemma}[theorem]{Lemma}
\newtheorem{proposition}[theorem]{Proposition}

\theoremstyle{definition}
\newtheorem{definition}[theorem]{Definition}
\newtheorem{example}[theorem]{Example}

\theoremstyle{remark}
\newtheorem{remark}[theorem]{Remark}

\usepackage[numbers]{natbib}
\title{Nested Sequents for Horn-Characterizable Quantified Modal Logics with Equality via Reachability Rules}
\author{Tim S. Lyon\\
\small Technische Universit{\"a}t Dresden\\
\small \texttt{timothy\_stephen.lyon@tu-dresden.de}
\and
Eugenio Orlandelli\\
\small University of Bologna\\
\small \texttt{eugenio.orlandelli@unibo.it}}
\date{}

\input{notation.tex}

\begin{document}

\maketitle
\thispagestyle{empty}

\noindent
\begin{minipage}{\textwidth}
\noindent
\textbf{Abstract.} We introduce cut-free nested sequent systems for a broad class of quantified modal logics (QMLs). The QMLs we consider are semantically defined using relational models that assign both an inner and outer domain to each world. This rich model structure enables the specification of various QMLs by enforcing different frame conditions, including increasing, decreasing, constant, and empty domains, as well as general path conditions and seriality. We extend the usual notion of nested sequent to include signatures, i.e., multisets of terms, which let us naturally define rules capturing the aforementioned domain conditions. A distinctive feature of our nested sequent systems is the use of reachability rules---inference rules parameterized by formal grammars (viz., semi-Thue systems). These rules operate by propagating or consuming formulae or terms along certain paths within a nested sequent, where paths are encoded as strings generated by a parameterizing grammar. This paper is the first to provide sound and complete nested systems for QMLs semantically characterized by models using both inner and outer domains. We analyze the proof-theoretic properties of these systems, identify a number of admissible structural rules, establish the invertibility of all rules, and prove a non-trivial syntactic cut-elimination theorem. We also observe that the standard universal quantifier rule used in nested systems subsumes the Extended Barcan Rule, which forces nested systems to capture QMLs with constant outer domains.
\end{minipage}

\section{Introduction}

\input{Intro2.tex}

\section{Preliminaries}\label{sec:log-prelims-I}

\subsection{Syntax and Semantics}\label{subsec:fo-modal-logics}

\input{body-prelims2.tex}

\subsection{Grammar-Theoretic Preliminaries}

\input{body-grammar-prelims2}

\section{Nested Sequent Systems}\label{sec:nested-calculi}

\input{body-nested-systems2.tex}

\section{Proof-Theoretic Properties}\label{sec:prop}

\input{body-proof-prop2.tex}

\section{Conclusion}\label{sec:conclusion}

\input{conclusion2.tex}

\bibliography{bibliography.bib}



\appendix

\input{app-proofs.tex}

\end{document}

%% file: notation.tex
\renewcommand{\phi}{\varphi}

\usepackage{mathtools}
\usepackage{synttree, bussproofs,scrextend, stmaryrd, rotating, xcolor}
\usepackage{amsmath, amssymb}
\usepackage{mathrsfs}
\usepackage{multicol}
\usepackage{comment}
\usepackage[all]{xy}
\usepackage{changepage,proof,upgreek,relsize,trfsigns}

\usepackage{algorithm}
\usepackage{algpseudocode}
\usepackage{wrapfig}

\usepackage{tikz}
\usetikzlibrary{positioning,arrows,calc,decorations.markings}
\tikzset{
modal/.style={>=stealth',shorten >=1pt,shorten <=1pt,auto,node distance=1.5cm,semithick},
world/.style={circle,draw,minimum size=0.5cm,fill=gray!15},
point/.style={circle,draw,inner sep=0.5mm,fill=black},
reflexive above/.style={->,loop,looseness=7,in=120,out=60},
reflexive below/.style={->,loop,looseness=7,in=240,out=300},
reflexive left/.style={->,loop,looseness=7,in=150,out=210},
reflexive right/.style={->,loop,looseness=7,in=30,out=330}
}

  \newtheorem{innercustomthm}{Theorem}
\newenvironment{customthm}[1]
  {\renewcommand\theinnercustomthm{#1}\innercustomthm}
  {\endinnercustomthm}

\newcommand{\iffi}{\textit{iff} }
\newcommand{\etc}{$\ldots$ }
\newcommand{\dfn}{Definition}
\newcommand{\thm}{Theorem}
\newcommand{\lem}{Lemma}

\newcommand{\sect}{Section}

\newcommand{\fig}{Figure}

\newcommand{\rmk}{Remark}

\newcommand{\albet}{\Upsigma}

\newcommand{\cate}{}
\newcommand{\sfour}{\mathrm{S4}}
\newcommand{\sfive}{\mathrm{S5}}
\newcommand{\langsfour}{L_{\sfour}}
\newcommand{\langsfive}{L_{\sfive}}

\newcommand{\concat}{\cdot}
\newcommand{\thuesys}{\mathrm{S}}
\newcommand{\g}{\mathrm{S}}

\newcommand{\glang}{L_{\g}}
\newcommand{\pto}{\longrightarrow}

\newcommand{\dr}{\pto^{*}_{\g}}
\newcommand{\fd}{R} 
\newcommand{\bd}{\bar{R}} 
\newcommand{\ques}{X} 
\newcommand{\albetstr}{\albet^{*}}
\newcommand{\stra}{\mathtt{s}}
\newcommand{\strb}{\mathtt{t}}
\newcommand{\conv}[1]{{#1}^{-1}}
\newcommand{\strc}{\mathtt{r}}
\newcommand{\empstr}{\varepsilon}

\newcommand{\R}{\mathcal{R}}

\newcommand{\names}{\mathit{Names}}

\newcommand{\var}{\mathit{Var}}
\newcommand{\con}{\mathit{Con}}
\newcommand{\ter}{\mathit{Ter}}
\newcommand{\rel}{\mathit{Rel}}

\newcommand{\imp}{\supset}

\newcommand{\lang}{\mathcal{L}} 

\newcommand{\baselog}{\mathsf{Q}^{\circ}_{=}.\mathsf{K}}
\newcommand{\ql}{\mathsf{Q}^{\circ}_{=}.\mathsf{K}(\fcons)}

\newcommand{\ns}{\mathcal{G}}
\newcommand{\nsii}{\mathcal{H}}
\newcommand{\nsiii}{\mathcal{K}}
\newcommand{\nant}{\Gamma}
\newcommand{\nantii}{\Sigma}
\newcommand{\ncon}{\Delta}

\newcommand{\hol}{\{}
\newcommand{\hor}{\}}

\newcommand{\prgr}[1]{\mathscr{P}\mathscr{G}(#1)}
\newcommand{\prv}{\mathscr{V}}
\newcommand{\pre}{\mathscr{E}}
\newcommand{\prl}{\mathscr{L}}

\newcommand{\ax}{\mathsf{ax}}

\newcommand{\disr}{\lor}
\newcommand{\conr}{\land}

\newcommand{\existsr}{\exists}
\newcommand{\allr}{\forall}

\newcommand{\wk}{\mathsf{w}}

\newcommand{\sym}{\mathsf{sym}_{\neq}}
\newcommand{\tra}{\mathsf{tra}_{\neq}}
\newcommand{\repgen}{\mathsf{grp}_{\neq}}
\newcommand{\cut}{\mathsf{cut}}
\newcommand{\cuti}{\mathsf{cut}_{=}}
\newcommand{\cuts}{\mathsf{cut}_{*}}

\newcommand{\ctr}{\mathsf{c}}

\newcommand{\ew}{\mathsf{ew}}
\newcommand{\ec}{\mathsf{ec}}

\newcommand{\ru}{\mathsf{r}}

\newcommand{\ps}{\mathsf{s}(t/x)}

\newcommand{\tw}{\mathsf{tw}}
\newcommand{\tc}{\mathsf{tc}}
\newcommand{\shift}{\mathsf{sft}(\fcons)}
\newcommand{\vs}{v(\fcset)}
\newcommand{\drule}{\mathsf{d}}

\newcommand{\true}{\mathtt{True}}

\newcommand{\prove}{\mathtt{Prove}}
\newcommand{\branch}{\mathcal{B}}

\newcommand{\D}{(d)}

\newcommand{\h}{\mathsf{H}}

\newcommand{\dia}{\diamondsuit} 
\renewcommand{\Diamond}{\diamondsuit}

\newcommand{\fcons}{\mathcal{C}}

\newcommand{\boxr}{\Box}
\newcommand{\diar}{\dia}

\newcommand{\nec}{\mathsf{nec}}

\newcommand{\dpr}{\mathsf{dp}}

\newcommand{\cdr}{\mathsf{cd}}
\newcommand{\nidref}{\mathsf{ref}_{\neq}}
\newcommand{\idrep}{\mathsf{rep}_{\neq}}
\newcommand{\drep}{\mathsf{rep}_{d}} 
\newcommand{\idrig}{\mathsf{rig}} 

\newcommand{\prf}{\pi}

\newcommand{\ned}{\mathsf{nd}}

\newcommand{\fcset}{\mathcal{C}}
\newcommand{\gpset}{\mathbf{G}}
\newcommand{\gpc}{\mathbf{G}}
\newcommand{\cser}{\mathbf{D}}
\newcommand{\cid}{\mathbf{ID}}
\newcommand{\cdd}{\mathbf{DD}}
\newcommand{\ccd}{\mathbf{CD}}
\newcommand{\cned}{\mathbf{NE}}
\newcommand{\acbf}{\mathbf{CBF}}
\newcommand{\abf}{\mathbf{BF}}
\newcommand{\aui}{\mathbf{UI}}

\newcommand{\ad}{\mathbf{D}}

\newcommand{\nql}{\mathsf{N}\ql}

\newcommand{\nqk}{\mathsf{NQ}^{\circ}_{=}.\mathsf{K}}
\newcommand{\I}[1]{I(#1)}

\def\s{\sigma}

\def\to{\supset}

\def\<{\langle}
\def\>{\rangle}
\def\to{\supset}

\def\ex{\exists}
\def\fa{\forall}
\newcommand{\E}{\mathcal{E}} 

\def\F{\mathcal{F}}
\def\M{\mathcal{M}}

\def\W{\mathcal{W}}
\def\W{\mathcal{W}}
\def\U{\mathcal{U}}
\def\D{\mathcal{D}}

\def\I{\mathcal{I}}
\def\R{\mathcal{R}}

\def\V{\mathcal{V}}

\def\fa{\forall}
\def\ex{\exists}

\def\vt{\vec{t}}
\def\vs{\vec{s}}
\def\vr{\vec{r}}

\newcommand{\simplepath}[1]{\overset{\mathclap{#1}}{\leadsto}}
\newcommand{\prpath}[1]{\hspace{3pt}\raise-3pt\hbox{$\simplepath{#1}$}\hspace{3pt}}

\newcommand{\fm}{\mathtt{fm}}
\newcommand{\tr}{\mathtt{tr}}

\newcommand{\ngn}[1]{\neg #1}
\newcommand{\negnnf}[1]{\dot{\neg} #1}

\newcommand{\msi}{\vec{\phi}} 
\newcommand{\fclass}{\mathscr{F}}
\newcommand{\fclassc}{\mathscr{F}(\fcons)}
\newcommand{\fv}[1]{FV(#1)}

\newcommand{\mint}{\iota}
\newcommand{\somesub}{/\!/}
\newcommand{\ebr}{\mathbf{EBR}}
\newcommand{\brn}[1]{\mathbf{BR}(#1)}
\newcommand{\set}[1]{\{#1\}}
\newcommand{\trir}{\triangleright}

\newcommand{\dotprf}[1]{\;\;\mathbin{\rotatebox[origin=c]{-90}{$\overset{\mathbin{\rotatebox[origin=c]{90}{$#1$}}}{\hdots}$}}}

%% file: Intro2.tex

Since Gentzen’s seminal work on proof theory~\cite{Gen35a,Gen35b}, \emph{analytic proof systems} have played a central role in both the study and application of logics. These systems (de)compose logical theorems step by step using inference rules, making them especially effective for establishing non-trivial properties such as interpolation and decidability, as well as for enabling automated reasoning. One of the most influential frameworks for constructing such systems is Gentzen’s \emph{sequent calculus}~\cite{Gen35a,Gen35b}, which represents proofs via (pairs of) multisets of formulae known as \emph{sequents}. Despite the utility of sequent systems, it is well-known that certain logics resist presentation as an analytic sequent calculus. This has led to numerous generalizations of Gentzen's formalism to supply ever more expressive logics with analytic proof systems. Such extensions are broadly called \emph{multisequent systems} and include \emph{display calculi}~\cite{Bel82}, \emph{hypersequent calculi}~\cite{Avr96}, \emph{labeled calculi}~\cite{Sim94,Vig00}, and \emph{nested calculi}~\cite{Bul92,Kas94}. This proof-theoretic paradigm has proven highly effective, being used to establish properties of logics such as decidability~\cite{Sim94,TiuIanGor12}, interpolation~\cite{FitKuz15,LyoKar24}, and complexity-hardness~\cite{LyoAlv22}. 

While multisequent systems for \emph{propositional} non-classical logics have been extensively studied, the literature on such systems for \emph{quantified} non-classical logics remains comparatively sparse. The quantified setting proves challenging because standard inference rules often prove too crude to capture the nuances inherent in non-classical quantifiers. For example, if we add the standard rules for the universal quantifier to a nested system for a quantified modal logic (QML), then both the Barcan formula ($\abf$) $\fa x \Box \phi \to \Box \fa x \phi$ and its converse ($\acbf$) $\Box\fa x\phi\to\fa x\Box \phi$ become provable, meaning, the nested system necessarily captures a QML with \emph{constant domains} (see \cite{LyoOrl23}). In the context of display calculi, Belnap has noted that this phenomenon occurs ``because these rules for the quantifiers are \emph{structure free} [$\ldots$], which is an indication of an unrefined account''~\cite[p.~409]{Bel82}. One can bypass this problem by incorporating \emph{signatures} (i.e., multisets of terms) into multisequents, which implicitly encode a (local) existence predicate $\E$ and restrict how quantifiers are introduced; see~\cite{Lyo21a,LyoOrl23,O21,Tiu11}. This allows for one to capture QMLs with alternative domain conditions such as increasing or decreasing domains (cf.~\cite{O21,LyoShiTiu25}).

In the context of QMLs, the labeled sequent formalism stands out as one of the most studied approaches, having been used to 
provide well-behaved calculi for a range of QMLs (e.g., \cite{O21,O24,NegPla11,Vig00}). As shown in~\cite{O24}, the labeled formalism enables the construction of (cut-free) sound and complete calculi for a broad class of semantically defined QMLs---many of which fail to have a known axiomatic system. Although the labeled sequent formalism is known to be quite general, uniform, and modular (see Lyon et al.~\cite{Lyoetal25} for a discussion), the formalism does have significant drawbacks. In particular, labeled sequent systems often introduce redundant syntactic structures in sequents and superfluous inferences in proofs, unnecessarily inflating proof size in many cases~\cite{Lyo25,LyoOst23}. Also, in contrast to other multisequent formalisms (e.g., hypersequents and nested sequents), labeled sequents typically fail to have a known formula interpretation, and in some cases, it is even known that labeled sequents fail to have a formula interpretation altogether~\cite{Lyoetal25}. To rectify these drawbacks, we study nested sequent systems for QMLs in this paper.

A \emph{nested sequent} is a tree of Gentzen sequents. The creation of the nested sequent formalism is often attributed to Bull~\cite{Bul92} and Kashima~\cite{Kas94}, though later work by Br\"unnler~\cite{Bru09} and Poggiolesi~\cite{Pog09} was instrumental to the popularity of the formalism.\footnote{It should be noted that Leivant~\cite[p. 361]{Lei81} introduced a notational variant of nested sequents in 1981 (which prefixes formulae with so-called \emph{execution sequences}) in his proof-theoretic work on propositional dynamic logic.} This formalism is rather elegant, yielding proof calculi that require minimal syntactic bureaucracy, have compact proofs, and where termination of proof search is more easily obtained~\cite{Lyo21thesis,LyoOst23}. 
Moreover, the rules of such systems are typically height-preserving invertible, which allows for counter-models to be extracted from failed proof search, and usually nested systems admit syntactic cut-elimination. Due to these nice aesthetic and computational properties, nested sequent systems have found a range of applications, being used in knowledge integration algorithms~\cite{LyoAlv22}, serving as a basis for constructive interpolation and decidability techniques~\cite{LyoTiuGorClo20,LyoKar24,TiuIanGor12}, and even being used to solve open questions about axiomatizability~\cite{IshKik07}. See~\cite{LelPog24} for a survey on nested sequent systems and their applications.

In this paper, we define and study nested sequent systems for a broad class of QMLs. The QMLs we consider subsume those treated by Corsi~\cite{Cor02}, and are semantically defined using relational models that assign both an inner and outer domain of elements to each world. Distinct QMLs are then characterized by imposing frame conditions on the accessibility relation of a model (as usual), or by imposing conditions on the inner domains associated with worlds (cf.~\cite{BraGhi07,Cor02,FitMen98}). Regarding the latter, one obtains distinct QMLs depending on if domains are permitted to be empty, if domains are permitted to increase, decrease, or be constant along the accessibility relation, or if inner and outer domains are required to be equal. The QMLs we consider consists of extensions of a base logic $\baselog$ with seriality, general path conditions, and the aforementioned domain conditions (cf.~\cite{Cor02}).
 
In order to give this class of logics a uniform nested sequent presentation, we use signatures (i.e., multisets of terms) in our nested sequents, which are used to process quantificational data. Additionally, we make use of \emph{reachability rules}, which possess two kinds of functionality: such rules may (1) propagate formulae or terms along paths in a nested sequent and/or (2) search for data along paths in a nested sequent. The first kind of functionality is exhibited by the well-known class of \emph{propagation rules}~\cite{Fit72,GorPosTiu11}, which have proven vital for the provision of nested sequent systems for propositional modal and constructive logics~\cite{CiaLyoRamTiu21,GorPosTiu08}. Rules exhibiting the second kind of functionality have only been defined more recently in the context of first-order constructive logics~\cite{Fit14,Lyo23,LyoShiTiu25}. We use these rules in our setting to capture reasoning with the diverse class of frame conditions we consider.

This paper serves as a journal version of the conference paper~\cite{LyoOrl23}, and makes the following new contributions:
\begin{itemize}

\item[(1)] We simplify our nested systems by using one-sided nested sequents for a language in negation normal form.

\item[(2)] We provide cut-free nested sequent calculi for a significant extension of the QMLs considered in~\cite{LyoOrl23}, permitting QMLs 
characterized by general path conditions of the form $\forall w, u, v \in\W (w \R^{n} u \ \& \ w \R^{k} v \rightarrow u \R v)$.

\item[(3)] We define and incorporate (new) reachability rules into our nested systems, which allow for a unified and modular  treatment of all QMLs within a single nested sequent presentation.

\item[(4)] We observe that the \emph{Extended Barcan Rule} $\ebr$ (see Corsi~\cite{Cor02}) is an instance of the universal quantifier rule in the nested sequent setting. This forces the QMLs we consider to have constant outer domains, suggesting that nested sequents `naturally capture' this class of QMLs, even with the use of signatures in nested sequents.

\item[(5)] We identify and prove a selection of structural rules height-preserving admissible in our nested systems, show that the rules of all nested systems are height-preserving invertible, and establish a non-trivial syntactic cut-elimination theorem. Our cut-elimination theorem is interesting in at least two respects: first, our use of reachability rules allows for the proof to be \emph{uniform} in the sense that \emph{ad hoc} elimination strategies are not needed to deal with special logics; cf.~\cite{Bru09,LyoOrl23}. Second, cut-elimination relies on the height-preserving admissibility of a novel structural rule, referred to as the \emph{shift rule}, which propagates nestings along paths of a nested sequent. The shift rule serves as a generic structural rule capturing all general path conditions imposed on a QML simultaneously.

\end{itemize}

\paragraph{Paper Organization.} In \sect~\ref{sec:log-prelims-I}, we define the class of QMLs we consider, as well as provide the grammar theoretical foundations for the formulation of our reachability rules. In \sect~\ref{sec:nested-calculi}, we introduce our nested sequent calculi for QMLs and prove them sound relative to the semantics used by Corsi~\cite{Cor02}. We also establish completeness by generalizing and adapting a method of Kripke~\cite{Kri59} to our nested sequent setting. To be more precise, we show how to extract counter-models from failed attempts to find a nested sequent proof of a given formula. In \sect~\ref{sec:prop}, we establish our height-preserving admissibility and invertibility results, which are used to prove the admissibility of the Extended Barcan Rule $\ebr$ in \emph{every} nested system we have defined. We also identify the novel `shift' structural rule that allows for a uniform proof of syntactic cut-elimination. In the final section (\sect~\ref{sec:conclusion}), we discuss future work and conclude.

%% file: body-prelims2.tex
We let $\var := \{x, y, z, \ldots\}$ and $\con := \{a, b, c, \ldots\}$ be a denumerable set of \emph{variables} and \emph{constants}, respectively, $\ter := \var \cup \con$ be the set of \emph{terms}, and $\rel := \{P^n,Q^n,R^n,\ldots\}$ be a set containing, for each $n \in \mathbb{N}$, a countable set of $n$-ary predicates, with \emph{propositional variables} being predicates of arity zero. We will often drop the superscript $n$ on $n$-ary predicates and write predicates as $P$, $Q$, $R$, $\ldots$ letting the context determine their arity. We use $t, r, s, \ldots$ to denote terms and often write a list of terms $t_{1}, \ldots, t_{n}$ as $\vec{t}$. We consider a \emph{quantified modal language} $\lang$ in negation normal form; it is the set of formulae generated via the following grammar in BNF:
$$
\phi ::= P(\vec{t}) \ | \ \ngn{P(\vec{t})} \ | \ t_1=t_2\ |\ t_1\neq t_2 \ | \ \phi \lor \phi \ | \ \phi \land \phi \ | \ \ex x\phi \ | \ \forall x \phi \ | \ \Diamond\phi \ | \ \Box \phi
$$

\noindent where $P \in \rel$ is an $n$-ary predicate, $\vec{t} = t_{1}, \ldots, t_{n} \in \ter$, and $x \in \var$. We use $\phi$, $\psi$, $\chi$, $\ldots$ to denote formulae from $\lang$. We define a \emph{literal} to be a formula of form $P(\vec{t})$, $\ngn{P(\vec{t})}$, $t = s$ or $t \neq s$, define a \emph{negative literal} to be a literal of the form $\ngn{P(\vec{t})}$ or $t \neq s$, and define an \emph{equality literal} to be a formula of the form $t = s$ or $t \neq s$. We will often use (annotated versions of) $L$, $N$, and $E$ to denote literals, negative literals, and equality literals, respectively. Given a literal $L$, we define its negation $\negnnf{L}$ as follows: (1) $\negnnf{P(\vec{t})} := \ngn{P(\vec{t})}$, (2) $\negnnf{\ngn{P(\vec{t})}} := P(\vec{t})$, (3) $\negnnf{(t = s)} := (t \neq s)$, and (4) $\negnnf{(t \neq s)} := (t = s)$. We define the negation $\negnnf{\phi}$ of an arbitrary formula in the usual way, and likewise, take the formulae $\bot$, $\top$, $\phi \to \psi$, and $\phi \equiv \psi$ to be defined as usual.

We use $\E t$ as short for $\ex x (x = t)$, and let $\E\vt$ denote $\bigwedge_{t\in\vt}\E t$. The \emph{length} of a formula is the number of symbols it contains. As usual, we say that the occurrence of a variable $x$ in $\phi$ is \emph{free} given that $x$ does not occur within the scope of a quantifier binding $x$, and we let $\fv{\phi}$ denote the set of all free variables occurring in $\phi$. We let $\phi(t/x)$ denote the formula obtained by substituting $t$ for each free occurrence of $x$ in $\phi$, and let $\phi(t \somesub x)$ denote a formula obtained by substituting $t$ for \emph{some} (\emph{all}, \emph{none}) free occurrences of $x$ in $\phi$. When performing substitutions we assume that we may possibly rename bound variables to avoid capture; e.g., if $y$ is distinct from $x$, then $(\forall y \phi)(y/x) \equiv \forall z(\phi(z/y)(y/x))$, where  $z$ is fresh. Moreover,  we identify formulas that differs only in the name of bound variables to avoid unnecessary bureaucracy.


As we are dealing with QMLs that include equality, we make use of specific kinds of relational models, called \emph{normal models}, in our semantics; cf.~Corsi~\cite[\dfn~2.2]{Cor02}. To allow for greater flexibility in characterizing logics, these models utilize two kinds of domains: \emph{outer domains} and \emph{inner domains}. The outer domain is taken to contain possible objects, which may or may not exist, whereas the inner domain is taken to contain existing objects. See~Corsi~\cite{Cor02} for a more in-depth discussion on the philosophical importance of these models.

\begin{definition}[Frame]\label{def:frame}
A \emph{frame} is a tuple $\F=\< \W,\,\R,\,\U,\,\D\>$ such that:
\begin{itemize}

\item $\W$ is a non-empty set $\{w, u, v, \ldots\}$ of \emph{worlds};

\item $\R \ \subseteq \W \times \W$ is a binary \emph{accessibility relation} on $\W$;

\item $\U$ is a non-empty set of objects called the \emph{outer domain};

\item $\D$ maps each $w \in \W$ to a set $\D_w$ (the \emph{inner domain of $w$}) such that $\D_w \subseteq \U$.

\end{itemize}
We say that a frame $\F$ has:
\begin{itemize}

\item \emph{increasing (inner) domains} \iffi for all $w,u \in \W$, $w\R u$ implies $\D_w \subseteq \D_u$;

\item \emph{decreasing (inner) domains} \iffi for all $w,u \in \W$, $w\R u$ implies $\D_u \subseteq \D_w$;

\item \emph{constant (inner) domains} \iffi for all $w,u \in \W$, $w\R u$ implies $\D_u = \D_w$;
\item \emph{classical domains} \iffi for all $w \in \W$, $\D_w = \U$;

\item \emph{non-empty (inner) domains} \iffi for each $w \in \W$, $\D_w \neq \emptyset$;

\item \emph{varying (inner) domains} \iffi none of the above conditions holds.

\end{itemize}
\end{definition}

\begin{definition}[Model]\label{def:model}
A \emph{model} $\M$ based on a frame $\F$ is an ordered pair $\<\F,\I\>$ where $\F$ is a frame and $\I$ is an \emph{interpretation} such that for each $w\in\W$, the following hold:
\begin{itemize}

\item For each $P^n \in \rel$, $\I_{w}(P^n) \subseteq \U^n := \overbrace{\U \times \cdots \times \U}^{n}$;

\item For each $a \in \con$, $\I_{w}(a) \in \U$,
where $w \R u$ implies $\I_{w}(a) = \I_{u}(a)$.
    
\end{itemize}
\end{definition}
We make the simplifying assumption that for each $w \in \W$, $\U^{0} = \{\<\>\}$ with $\<\>$ the empty tuple, meaning $\I_{w}(P) = \{\<\>\}$ or $\I_{w}(P) = \emptyset$, for any  propositional variable, i.e. 0-ary predicate $P$.

\begin{remark}\label{rmk:constant-outer} Corsi's original (normal) models do not make use of a single outer domain $\U$ as in \dfn~\ref{def:frame} above, but rather, associate an outer domain with each world and are permitted to grow along the accessibility relation (see~\cite[\dfn~2.2]{Cor02}). As will be discussed in \sect~\ref{sec:prop}, the nested sequent calculi we introduce include a generalization of the \emph{Extended Barcan Rule $\ebr$}, where $\ebr$ is defined to be the set of all rules $\brn{n+1}$ for $n \in \mathbb{N}$:

\begin{center}
\AxiomC{$\phi_{0} \imp \Box (\phi_{1} \imp \cdots \imp \Box (\phi_{n} \imp \Box\phi_{n+1})\ldots)$}
\RightLabel{$\brn{n+1}$~\text{ with }~$x \not\in \fv{\phi_{0}, \ldots, \phi_{n}}$}
\UnaryInfC{$\phi_{0} \imp \Box (\phi_{1} \imp \cdots \imp \Box (\phi_{n} \imp \Box\forall x \phi_{n+1})\ldots)$}
\DisplayProof
\end{center}

As pointed out by Corsi~\cite[pp. 1503-1504]{Cor02}, the inclusion of $\ebr$ in a QML forces the outer domains of normal models to be \emph{constant}. Therefore, we have defined our models to have constant outer domains from the onset since the admissibility of $\ebr$ in our nested sequent calculi (see Theorem \ref{theoremEBR}) forces the QMLs they capture to have constant outer domains.
\end{remark}

Given a model $\M = \<\W,\R,\U,\D,\I\>$ and some $w \in \W$, an \emph{$\M$-assignment} is a function $\s : \var \rightarrow \U$ mapping variables to objects in the outer domain $\U$. 
We let $\sigma^{ x \trir o}$ be the $\M$-assignment mapping $x$ to $o \in \U$ and behaving like the $\M$-assignment $\sigma$ for all other variables. Given an $\M$-assignment $\s$, we can interpret all terms of the language by letting $\I_{w}^{\s}(a) := \I_{w}(a)$ for $a \in \con$ and $\I_{w}^{\s}(x) := \s(x)$ for $x \in \var$. For a list of terms $\vt = t_{1}, \ldots, t_{n}$, we define $\I_{w}^{\s}(\vt) := \I_{w}^{\s}(t_{1}), \ldots, \I_{w}^{\s}(t_{n})$.

\begin{remark} In Corsi's original normal models, assignments are defined relative to specific worlds, that is to say, since each world $w$ is associated with its own outer domain $\U_w$, each assignment interprets variables locally in the outer domain $\U_w$ (see \cite[p.~1485]{Cor02}). As discussed in \rmk~\ref{rmk:constant-outer} above, our nested sequent calculi naturally characterize logics that have constant outer domains. In this setting, Corsi's local (world-dependent) assignments become global as they interpret variables on the singular outer domain $\U$.
\end{remark}

\begin{definition}[Semantic Clauses]
Let $\M = \<\W,\R,\U, \D,\I\>$ be a model with $w \in \W$ and let $\s$ be an $\M$-assignment. The \emph{satisfaction relation} $\Vdash$ is defined as follows:\bigskip

\begin{tabular}{lll}
$\M,w,\sigma\Vdash P(t_1,\dots t_n)$\quad\phantom{a}&iff\quad\phantom{a}&$\<\I^{\s}(t_1),\dots,\I^{\s}(t_n)\>\in\I_{w}(P)$\\\noalign{\smallskip}

$\M,w,\sigma\Vdash \ngn{P(t_1,\dots t_n)}$\quad\phantom{a}&iff\quad\phantom{a}&$\<\I^{\s}(t_1),\dots,\I^{\s}(t_n)\>\not\in\I_{w}(P)$\\

$\M,w,\sigma\Vdash t = s$\quad\phantom{a}&iff\quad\phantom{a}&$\I_{w}^{\s}(t)=\I_{w}^{\s}(s)$\\

$\M,w,\sigma\Vdash t \neq s$\quad\phantom{a}&iff\quad\phantom{a}&$\I_{w}^{\s}(t) \neq \I_{w}^{\s}(s)$\\

$\M,w,\sigma\Vdash \phi\lor\psi$\quad\phantom{a}&iff\quad\phantom{a}&$\M,w,\sigma\Vdash\phi$ or $\M,w,\sigma\Vdash\psi$\\

$\M,w,\sigma\Vdash \phi\wedge\psi$\quad\phantom{a}&iff\quad\phantom{a}&$\M,w,\sigma\Vdash \phi$ and $\M,w,\sigma\Vdash \psi$\\

$\M,w,\sigma\Vdash \ex x\phi$&iff& for some $o\in\D_w$, $\M,w,\sigma^{x\trir o}\Vdash\phi$\\

$\M,w,\sigma\Vdash \fa x\phi$&iff& for all $o\in\D_w$, $\M,w,\sigma^{x\trir o}\Vdash\phi$\\

$\M,w,\sigma\Vdash\Diamond\phi$&iff&
for some $u\in\W$, $w\R u$ and $\M,u,\sigma\Vdash\phi$\\

$\M,w,\sigma\Vdash\Box\phi$&iff&
for all $u\in\W$, if $w\R u$, then $\M,u,\sigma\Vdash\phi$\\
\end{tabular}\medskip

\noindent A formula $\phi$ is \emph{true at a world} $w$ of a model $\M$ \iffi $\M,w,\s \Vdash \phi$ for every assignment $\s$ over $\M$. A formula $\phi$ is \emph{true in a model} $\M$ \iffi it is true at all worlds of that model; a formula $\phi$ is \emph{false in a model} $\M$ \iffi it is not true at some world of that model. A formula $\phi$ \emph{valid relative to a class of frames} $\fclass$ \iffi it is true at all models based on a frame in that class.
\end{definition}

\begin{remark} Observe that if $\M, w, \sigma \Vdash \E t$, then $\sigma(t) \in \D_{w}$, that is, if the existence predicate holds of a term $t$ at a world $w$, then its interpretation is in the \emph{inner domain} of $w$.
\end{remark}

We say that a model $\M = \<\W,\R,\U, \D,\I\>$ is \emph{connected} \iffi for any two worlds $w, u \in \W$ there exists a (potentially undirected) $\R$-path connecting $w$ to $u$. The following proposition will be useful in the sequel.

\begin{proposition}\label{prop:eq-sat-global} The following two claims hold:
\begin{itemize}

\item[(1)] A formula is false on a model \iffi it is false on a connected model.

\item[(2)] Let $\M = \<\W,\R,\U, \D,\I\>$ be a connected model with $w \in \W$, let $\s$ be an $\M$-assignment, and let $E$ be an equality atom. Then, $\M, \s, w \Vdash E$ \iffi for all $u \in \W$, $\M, \s, u \Vdash E$.

\end{itemize}
\end{proposition}

\subsection{Quantified Modal Logics}

\begin{figure}[t]
\framebox{
\begin{minipage}{0.47\textwidth}
\begin{description}
\item[TAUT] Classical propositional tautologies 
\item[K] $\Box (\phi\to \psi)\to(\Box \phi\to \Box \psi)$
\item[UI$^\circ$] $\fa y(\fa x\phi\to \phi(y/x))$
\item[$\fa$-COMM] $\fa x\fa y\phi\to\fa y\fa x\phi$ 
\item[$\fa$-DIST] $\fa x(\phi\to \psi)\to(\fa x\phi\to\fa x \psi)$
\item[$\fa$-VAQ] $ \phi\to \fa x\phi$, if $x$ is not free in $\phi$
\end{description}
\end{minipage}
\begin{minipage}{0.49\textwidth}
\begin{description}
\item[REF] $t=t$
\item[REPL] $t=s\wedge \phi(t \somesub x)\to \phi(s \somesub x)$
\item[ND] $t\neq s\to \Box (t\neq s)$
\item[EBR] See \rmk~\ref{rmk:constant-outer}
\item[MP] If $\phi$ and $\phi\to \psi$ are theorems so is $\psi$
\item[N] If $\phi$ is a theorem so is $\Box \phi$ 
\item[UG] If $\phi$ is a theorem so is $\fa x \phi$
\end{description}
\end{minipage}\smallskip
}
\caption{Axiom System $\h \baselog$ for $\baselog$.}\label{fig:axioms-baselog}
\end{figure}

The base logic $\baselog$ is defined to be the set of all formulae valid relative to the class of all frames, which is equivalently characterized by the axiom system $\h \baselog$ given in \fig~\ref{fig:axioms-baselog}.\footnote{We use the term \emph{axiom} to mean \emph{axiom schema} in this paper.} In this paper, we also consider extensions of $\baselog$ obtained by enforcing the \emph{frame conditions} shown in \fig~\ref{fig:axioms-related-conditions}. Following the top-down order of \fig~\ref{fig:axioms-related-conditions}, such frame conditions include: 
\begin{itemize}

\item \emph{seriality} $\ad := \fa w \in \W \,\ex u\in\W(w\R u)$;

\item \emph{generalized path conditions} $\gpc(n,k) := \forall w, u, v \in\W (w \R^{n} u \ \& \ w \R^{k} v \rightarrow u \R v)$; 

\item \emph{increasing (inner) domains} $\cid := \forall w,v \in\W(w\R v \rightarrow \D_w\subseteq \D_v)$;

\item \emph{decreasing (inner) domains} $\cdd := \forall w,v \in\W(w\R v \rightarrow \D_w\supseteq \D_v)$; 

\item \emph{classical domains} $\ccd := \fa w \in \W(\D_w = \U)$;

\item \emph{non-empty (inner) domains} $\cned := \forall w \in \W \,\exists o \in D_{w}$. 
    
\end{itemize}
We let $n,k \in \mathbb{N}$, use $\gpc(n,k)$ to denote a generalized path condition, and $\gpset$ to denote a \emph{set} of such conditions. These conditions cover well-known frame conditions such as reflexivity (when $n = k = 0$), transitivity (when $k = 2$ and $n= 0$), symmetry (when $n = 1$ and $k = 0$), and Euclideanity (when $n = k = 1$). Also, observe that constant (inner) domains can be enforced on a frame by enforcing both $\cid$ and $\cdd$. For a set $\fcons$ of frame conditions from \fig~\ref{fig:axioms-related-conditions}, we define $\fclassc$ to be the class of frames satisfying the frame conditions in $\fcons$. For a set $\fcons$ of frame conditions, we define the \emph{quantified modal logic (QML)} $\ql$ to be the the set of all formulae valid relative to the class $\fclassc$ of frames. (NB. Observe that $\baselog(\emptyset) = \baselog$.)

\begin{remark}\label{rmk:axs-closed-under-consequence} We make the simplifying assumption that if $\fcons$ is a set of frame conditions and a frame condition $\mathbf{C}$ holds on every frame in $\fclassc$, then $\mathbf{C} \in \fcons$. 
For example, if $\ccd \in \fcons$, then $\cid, \cdd, \cned \in \fcons$ since every frame satisfying $\ccd$ also satisfies the latter three conditions.
\end{remark}

\begin{figure}[t]
\begin{center}
\framebox{
\bgroup
\setlength{\tabcolsep}{5pt}
\def\arraystretch{1}
\begin{tabular}{c | c | c}
Name & Frame Condition & Corresponding Axiom\\\noalign{\smallskip}\hline\hline\noalign{\smallskip}

$\cser$ & $\fa w \in \W \,\ex u\in\W(w\R u)$ & $\Box \phi \to\Diamond \phi$\\\noalign{\smallskip}

\,$\gpc(n,k)$ & $\forall w, u, v \in\W (w \R^{n} u \ \& \ w \R^{k} v \rightarrow u \R v)$ & $\dia^{k} \phi \to \Box^{n} \dia \phi$\\\noalign{\smallskip}

$\cid$ & $\forall w,v \in\W(w\R v \rightarrow \D_w\subseteq \D_v)$ & $\Box\fa x\phi\to\fa x\Box \phi$\\\noalign{\smallskip}

$\cdd$ & $\forall w,v \in\W(w\R v \rightarrow \D_w\supseteq \D_v)$&$\fa x\Box \phi\to\Box\fa x\phi$\\\noalign{\smallskip}

$\ccd$ & $\fa w \in \W(\D_w = \U)$ & $\fa x \phi \to \phi(t/x)$\\\noalign{\smallskip}

$\cned$ & $\forall w \in \W \,\exists o \in D_{w}$ & $\forall x \phi \to \exists x \phi$\\
\end{tabular}
\egroup
}
\end{center}
\caption{Additional frame/domain conditions and their corresponding axioms. We note that when $n=0$, the second frame condition is $\forall w, v (w \R^{k} v \rightarrow w \R v)$, when $k = 0$, the frame condition is $\forall w, u (w \R^{n} u \rightarrow u \R w)$, and when $n = k = 0$, the frame condition is $\forall w (w \R w)$.
}
\label{fig:axioms-related-conditions}
\end{figure}

As mentioned above, the axiom system $\h \baselog$ characterizes the base logic $\baselog$ (see \fig~\ref{fig:axioms-baselog}). We remark that $\ebr$ can be dropped from the axiom system; as proven by Corsi~\cite[\lem~2.13]{Cor02}, the logic $\baselog$ is equivalently characterized by $\h \baselog$ and $\h \baselog$ \emph{without} $\ebr$. Nevertheless, we include $\ebr$ in the axiomatization of $\baselog$ to highlight that $\ebr$ belongs to all logics we are considering. 

It is well-known that certain axioms are canonical for the various frame conditions discussed above. As shown in \fig~\ref{fig:axioms-related-conditions} (reading the table top-down), each frame condition corresponds to one of the following axioms:
\begin{itemize}

\item \emph{seriality axiom} $\Box \phi \imp \Diamond \phi$;

\item \emph{generalized path axioms} $\dia^{k} \phi \to \Box^{n} \dia \phi$ with $n,k \in \mathbb{N}$;

\item \emph{Converse Barcan Formula} $\acbf := \Box\fa x\phi\supset\fa x\Box \phi$;

\item \emph{Barcan Formula} $\abf := \fa x\Box \phi\supset\Box\fa x \phi$;

\item \emph{universal instantiation axiom} $\aui := \forall x \phi \imp \phi(t/x)$;

\item \emph{non-empty domain axiom} $\forall x \phi \imp \exists x \phi$.
    
\end{itemize}
As far as we know, over frames without   classical (and constant) domains there are soundness and completeness results for particular axiom systems (see~\cite{Cor02,FitMen98,TanakaOno}), but no result covering all QMLs considered here. In \sect~\ref{sec:nested-calculi}, we will provide a sound and complete (cut-free) nested sequent system for each QML $\ql$. Nevertheless, we list some noteworthy soundness and completeness results below, which can found in Corsi~\cite{Cor02}:
\begin{itemize}

\item $\h \baselog$ is sound and complete w.r.t. $\fclass(\emptyset)$; 

\item $\h \baselog \cup \set{\acbf}$ is sound and complete w.r.t. $\fclass(\set{\cid})$; 

\item $\h \baselog \cup \set{\abf}$ is sound and complete w.r.t. $\fclass(\set{\cdd})$; 

\item $\h \baselog \cup \set{\acbf, \abf}$ is sound and complete w.r.t. $\fclass(\set{\cid,\cdd})$; 

\item $\h \baselog \cup \set{\aui, \abf}$ is sound and complete w.r.t. $\fclass(\set{\ccd, \cid,\cdd,\cned})$. 

\end{itemize}


%% file: body-grammar-prelims2.tex
As will be seen later on, a significant aspect of our nested calculi is the incorporation of inference (viz. reachability) rules whose applicability depends upon certain strings generated by a formal grammar. Therefore, the current section introduces the grammar-theoretic notions required to properly define such rules.

We define our \emph{alphabet} $\albet := \{\fd,\bd\}$ to be the set of \emph{characters} $\fd$ and $\bd$, and we let $\ques \in \albet$. We make use of the symbols $\fd$ and $\bd$ to encode information in inference rules about worlds in the \emph{future} and \emph{past} (resp.) of an accessibility relation $\R$ in a model. We say that $\fd$ and $\bd$ are \emph{converses}, and define the \emph{converse operation} $\conv{\ques}$ (on characters) as $\conv{\fd} := \bd$ and $\conv{\bd} := \fd$.

We let $\concat$ denote the typical \emph{concatenation operation} with $\varepsilon$ the \emph{empty string}. The set $\albet^{*}$ of \emph{strings over $\albet$} is defined to be the smallest set satisfying the following conditions: (i) $\albet \cup \{\varepsilon\} \subseteq \albet^{*}$, and (ii) $\text{If } \stra \in \albet^{*} \text{ and } \ques \in \albet \text{, then } \stra \concat \ques \in \albet^{*}$. We use $\stra$, $\strb$, $\strc$, \etc (potentially annotated) to denote strings from $\albetstr$. Also, we will often write the concatenation of two strings $\stra$ and $\strb$ as $\stra \cate \strb$ as opposed to $\stra \concat \strb$, and for the empty string $\empstr$, we have $\stra \cate \empstr = \empstr \cate \stra = \stra$. The converse operation on strings (adapted from~\cite{TiuIanGor12}) is defined accordingly: (i) $\conv{\varepsilon} := \varepsilon$, and (ii) $\text{If } \stra = \ques_{1} \cdots \ques_{n} \text{, then } \conv{\stra} := \conv{\ques}_{n} \cdots \conv{\ques}_{1}$.

We now define \emph{$\albet$-systems}, which are special kinds of \emph{semi-Thue systems}~\cite{Pos47} that encode information about the conditions imposed on frames and models, and which will parameterize inference rules in our proof systems later on.

\begin{definition}[$\albet$-System]\label{def:grammar} We define a \emph{$\albet$-system} to be a set $\g$ of \emph{production rules} of the form $\ques \pto \stra$, where $\ques \in \albet$ and $\stra \in \albetstr$.
\end{definition}

Specific $\albet$-systems will be of use to us in this paper. In particular, for a set $\gpset$ of generalized path conditions, we define the $\albet$-system $\g(\gpset)$ as follows: 
\begin{center}
$(\fd \pto \bd^{n} \cate \fd^{k}), (\bd \pto \bd^{k} \cate \fd^{n}) \in \g(\gpset)$ \iffi $\gpc(n,k) \in \gpset$.
\end{center}

\noindent
We also define the specific $\albet$-systems $\sfour$ and $\sfive$:
\begin{itemize}

\item $\sfour := \{\fd \pto \empstr, \bd \pto \empstr, \fd \pto \fd \fd, \bd \pto \bd \bd\}$;

\item $\sfive := \{\fd \pto \empstr, \bd \pto \empstr, \fd \pto \bd \fd, \bd \pto \bd \fd\}$.

\end{itemize}
Since $\sfour$ and $\sfive$ are $\albet$-systems, any property that holds for an arbitrary $\albet$-system holds for $\sfour$ and $\sfive$ as well. We will describe the significance of these two $\albet$-systems shortly. First, let us make the notion of a \emph{derivation} and \emph{language} precise in the context of $\albet$-systems.



\begin{definition}[Derivation, Language]\label{def:semi-thue-deriv-lang}
 Let $\g$ be a $\albet$-system. We write $\stra \pto \strb$ and say that  the string $\strb$ may be derived from the string $\stra$ in $\albetstr$ in \emph{one-step} \iffi there are strings $\stra', \strb' \in \albetstr$ and $\ques \pto \strc \in \thuesys$ such that $\stra = \stra' \cate \ques \cate \strb'$ and $\strb = \stra' \cate \strc \cate \strb'$. We define the \emph{derivation relation} $\dr$ to be the reflexive and transitive closure of $\pto$. For $\stra, \strb \in \albetstr$, we call $\stra \dr \strb$ a \emph{derivation of $\strb$ from $\stra$}, and define the \emph{length} of a derivation to be the minimal number of one-step derivations required to derive $\strb$ from $\stra$ in $\g$. Moreover, we define the \emph{language} $\glang(\stra) := \{\strb \ | \ \stra \dr \strb \}$, where $\stra \in \albet^{*}$. 
 \end{definition}

\begin{lemma}\label{lem:S4-with-S(G)-is-S4-or-S5} Let $\gpset$ be a set of generalized path conditions. 
\begin{itemize}

\item[(1)] $\stra = \ques_{1} \cdots \ques_{n} \in L_{\g(\gpset)}(\ques)$ \iffi $\conv{\stra} = \conv{\ques}_{n} \cdots \conv{\ques}_{1} \in L_{\g(\gpset)}(\conv{\ques})$. 
    
\item[(2)] If $L = L_{\sfour \cup \g(\gpset)}(\ques)$, then either $L = L_{\sfour}(\ques)$ or $L = L_{\sfive}(\ques)$.

\end{itemize}
\end{lemma}

\begin{proof} Claim (1) follows from the definition of $\g(\gpset)$ since for each production rule $\ques \pto \stra \in \g(\gpset)$ there is a production rule $\conv{\ques} \pto \conv{\stra} \in \g(\gpset)$. To prove claim (2), we first observe that $\g(\gpset)$ either contains a production rule of the form $\ques \pto \stra \conv{\ques} \strb$ or it does not. Since $\fd \pto \empstr, \bd \pto \empstr \in \sfour$, we have that $\ques \longrightarrow_{\sfour \cup \g(\gpset)}^{*} \conv{\ques}$, which, along with the other rules of $\sfour$, will permit us to derive any string from $\ques$, i.e. $L = L_{\sfive}(\ques)$. If $\g(\gpset)$ does not contain a production rule of the form $\ques \pto \stra \conv{\ques} \strb$, then it will also not contain a production rule of the form $\conv{\ques} \pto \strc \ques \strc'$ by the same reason justifying claim (1) above. Therefore, every production rule of $\g(\gpset)$ must be of the form $\fd \pto \fd^{n}$ or $\bd \pto \bd^{n}$ with $n \in \mathbb{N}$ and $\ques^{0} = \empstr$. Observe that all production rules in $\g(\gpset)$ can be derived in $\sfour$, i.e. $\fd \longrightarrow_{\sfour}^{*} \fd^{n}$ or $\bd \longrightarrow_{\sfour}^{*} \bd^{n}$, showing that $L = L_{\sfour}(\ques)$ as $\sfour$ can already derive whatever can be derived with $\g(\gpset)$.
\end{proof}
 
$\sfour$ encodes a notion of \emph{directed} reachability and $\sfive$ encodes a notion of \emph{undirected} reachability. Given a frame, recognizing what worlds are reachable from one another along paths of the accessibility relation $\R$ is relevant in our context as frame conditions such as $\gpc(n,k)$ connect initial and terminal worlds, and $\cid$ and $\cdd$ require inner domains to contain certain elements along $\R$-paths. To see how $\sfour$ and $\sfive$ encode notions of (un)directed reachability, let us think of $\fd$ as representing a forward move along the accessibility relation of a model, and $\bd$ as representing a backward move along the accessibility relation. Then, $\sfour$ defines the languages $\langsfour(\fd) = \{\fd^{n} \ | \ n \in \mathbb{N}\}$ and $\langsfour(\bd) = \{\bd^{n} \ | \ n \in \mathbb{N}\}$, where each string $\fd^{n}$ and $\bd^{n}$ (denoting a sequence of $n$ $\fd$'s and $n$ $\bd$'s, respectively, with $\fd^{0} = \bd^{0} = \empstr$) `connects' a world $w$ and $u$ if $u$ can be reached in $n$ forward steps or backward steps (resp.) along the accessibility relation from $w$. Similarly, $\sfive$ defines the language $\langsfive(\fd) = \langsfive(\bd) = \albetstr$ where each string $\ques_{1} \cdots \ques_{n}$ `connects' a world $w$ to $u$ if $u$ is reachable from $w$ in $n$ steps (forward and/or backward).




%% file: body-nested-systems2.tex

We now present our nested sequent systems for QMLs. Let $\vec{\phi}$ be a (finite) multiset of formulae and $\vt \subset \ter$ be  a (finite) multiset of terms, which we refer to as a \emph{signature}. We define \emph{nested sequents} accordingly:
\begin{enumerate}

\item Each \emph{flat sequent} of the form $\vt, \vec{\phi}$ is a nested sequent;

\item For $1 \leq i \leq n$, if $\nsii_{i}$ is a nested sequent, then $\vt, \vec{\phi}, [\nsii_{1}], \ldots, [\nsii_{n}]$ is a nested sequent.

\end{enumerate}
In the context of nested sequents, we use $\Gamma$, $\Delta$, $\nantii$, $\ldots$ to denote flat sequents, $\vt$, $\vs$, $\vr$, $\ldots$ to denote signatures, $\vec{\phi}$, $\vec{\psi}$, $\vec{\xi}$, $\ldots$ to denote multisets of formulae, and $\ns$, $\nsii$, $\nsiii$, $\ldots$ to denote nested sequents. The use of signatures in nested sequents is crucial for encoding the $\cid$, $\cdd$, $\ccd$, and $\cned$ conditions in inference rules. Signatures have been used in nested sequents for other first-order non-classical logics~\cite{Lyo23,LyoShiTiu25} to encode similar conditions.

For a nested sequent $\ns =  \Gamma, [\nsii_{1}], \ldots, [\nsii_{m}]$, we define a \emph{component} of the nested sequent to be a flat sequent appearing in $\ns$, that is, a component of $\ns$ is an element of the multiset $c(\ns)$, where the function $c$ is recursively defined as follows with $\uplus$ denoting the multiset union:
$$
c( \Gamma, [\nsii_{1}], \ldots, [\nsii_{m}]) = \{ \Gamma \} \uplus \!\! \biguplus_{1 \leq i \leq m} \!\! c(\nsii_{i}).
$$
We let $\names$ be a denumerable set of pairwise distinct labels and use $w$, $u$, $v$, $\ldots$ (occasionally annotated) to denote them. For each nested sequent $\ns = \Gamma , [\nsii_{1}], \ldots, [\nsii_{m}]$, we call $\Gamma$ the \emph{root} of $\ns$ and assume that every component of $\ns$ is assigned a unique name from $\names$. The incorporation of names in our nested sequents is crucial for the definition of our reachability rules below. A \emph{context} is a nested sequent with holes; a hole $\{\}$ takes the place of a formula in a nested sequent and $\ns\{\nsii\}$ denotes the sequent obtained by filling the hole in the context $\ns\{\}$ with the nested sequent $\nsii$. As usual, we write $\ns\{\emptyset\}$ to denote the removal of a hole $\{\}$ in a context $\ns\{\}$. For example, if $\ns\set{} = t, p(x), [s, q(a), \set{}]$ and $\nsii = x, r(x,y), [\exists x p(x)]$, then $\ns\set{\nsii} = t, p(x), [s, q(a), x, r(x,y), [\exists x p(x)]]$ and $\ns\set{\emptyset} = t, p(x), [s, q(a)]$. We sometimes refer to a component as a \emph{$w$-component} if $w$ is the name of that component, and we use the notation $\ns\{\nsii_{1}\}_{w_{1}}\cdots\{\nsii_{n}\}_{w_{n}}$ to indicate that $\nsii_{i}$ is `rooted at' the $w_{i}$-component of $\ns$ for $1 \leq i \leq n$. Often, we will more simply write $\ns\{\nsii_{1}\}\cdots\{\nsii_{n}\}$ if the names of the components are not of relevance. 
Observe that every nested sequent encodes a \emph{tree} whose nodes are (named) flat sequents.

\begin{definition}[Tree of a Nested Sequent]\label{definition:tns} Let $\ns =  \Gamma, [\nsii_{1}]_{w_1}, \ldots, [\nsii_{m}]_{w_n}$ be a nested sequent with $w$ the name of the component $\Gamma$. We recursively define the tree $\tr(\ns) = (V,E)$ as follows:
$$
V = \{(w,\Gamma)\} \cup \bigcup_{i = 1}^{n} V_{i}
\qquad
E = \{(w,w_{i}) \ | \ 1 \leq i \leq n\} \cup \bigcup_{i = 1}^{n} E_{i}
$$
such that $\tr(\nsii_{i}) = (V_{i},E_{i})$ for $1 \leq i \leq n$.
\end{definition}


As is typical of nested sequents, they admit interpretation directly on models \emph{and} possess an equivalent formula interpretation, that is, every nested sequent is equivalent to a formula in the language $\lang$. Both interpretations are useful, and so, we present both interpretations below. 

\begin{definition}[Sequent Semantics]\label{def:sequent-semantics} Let $\ns$ be a nested sequent with $w$ the name of the root, $\tr(\ns) = (V,E)$, $\M = \<\W,\R,\U,\D,\I\>$ be a model, and $\s$ be an $\M$-assignment. We define an \emph{$\M$-interpretation} to be a function $\mint : \names \rightarrow \W$. We say that $\ns$ is satisfied on $\M$ with $\s$ and $\mint$, written $\M, \s, \mint \Vdash \ns$, \iffi if conditions (1) and (2) hold, then (3) holds, where conditions (1)--(3) are listed below:
\begin{itemize}

\item[(1)] for each $(u,v) \in E$, $\mint(u)\R\mint(v)$;

\item[(2)] for each $(u,\vt, \msi) \in V$, $\I_{\mint(u)}^{\s}(\vt) \in \D_{\mint(u)}$;

\item[(3)] for some $(v,\vt, \msi) \in V$, $\M, \mint(v), \s \Vdash \bigvee \msi$.

\end{itemize}
We say that $\ns$ is \emph{valid} w.r.t. a class of frames $\fclassc$ \iffi for every model $\M$ based on a frame $\F \in \fclassc$, $\M$-assignment $\s$, and $\M$-interpretation $\mint$, $\M, \s, \mint \Vdash \ns$. We say that $\ns$ is \emph{invalid} otherwise.
\end{definition}

\begin{definition}[Formula Interpretation] We define the formula interpretation $\fm(\ns)$ of a nested sequent $\ns$ accordingly:
$$
\fm(\vec{t},\vec{\phi},[\nsii_1],\dots,[\nsii_n]) :=(\bigvee_{t\in\vec{t}}\negnnf{\E t}\lor\bigvee_{\phi\in\vec{\phi}}\phi) \lor \bigvee_{i=1}^{n}\Box \,\fm(\nsii_{i})
$$
where $\bigvee \emptyset := \bot$, as usual.
\end{definition}

Observe that $\fm(\ns)$ requires the existence predicate $\E$ to be expressible in the object language, which holds when the language contains the equality, as is the case for $\lang$.

\begin{proposition}\label{prop:equiv-model-form-interp}
Let $\M = \<\W,\R,\U,\D,\I\>$ be a model, $\s$ be an $\M$-assignment, $\mint$ be an $\M$-interpretation, and $\ns = \vec{t},\vec{\phi},[\nsii_1],\dots,[\nsii_n]$ be a nested sequent with $w$ the name of the root $\vec{t},\vec{\phi}$. Then, $\M, \s, \mint \Vdash \ns$ \iffi $\M, \mint(w), \s \Vdash \fm(\ns)$.
\end{proposition}

A uniform presentation of our nested sequent systems is provided in \fig~\ref{fig:nested-calculi} and all systems are formally defined in \dfn~\ref{def:nested-systems} (after enough groundwork has been laid to define them properly). Each nested sequent system $\nql$ takes a set $\fcons$ of frame conditions as a parameter and is sound and complete for the logic $\ql$ (see Theorems~\ref{thm:soundness-nested} and~\ref{thm:completeness-nested} below). We note that  $\nqk = \nqk(\emptyset)$ is the calculus for $\baselog$. Each system contains the \emph{initial rule} $\ax$, the \emph{logical rules} $\disr$, $\conr$, $\existsr$, $\allr$, $\diar$, and $\boxr$, and the \emph{identity rules} $\nidref$, $\idrep$, $\idrig$, and $\drep$. The rules $\dpr$, $\ned$, and $\cdr$ are referred to as \emph{signature rules} and the $\drule$ rule is the only \emph{structural rule}. (NB. The side conditions of the rules are discussed in detail below.) Recall that $L$ stands for a literal and  $N$ stands for a negative literal,  which is relevant for reading the $\ax$ and $\idrep$ rules.

\begin{figure}[t]

\begin{center}
\begin{tabular}{c c c c}
\AxiomC{\phantom{$\ns$}}
\RightLabel{$\ax$}
\UnaryInfC{$\ns \hol  L, \negnnf{L} \hor$}
\DisplayProof

&

\AxiomC{$\ns \hol \phi,\psi\hor$}
\RightLabel{$\disr$}
\UnaryInfC{$\ns \hol \phi\lor\psi \hor$}
\DisplayProof

&

\AxiomC{$\ns \hol \phi\hor$}
\AxiomC{$\ns \hol \psi \hor$}
\RightLabel{$\conr$}
\BinaryInfC{$\ns \hol \phi \wedge \psi  \hor$}
\DisplayProof

&

\AxiomC{$\ns \hol t, \ex x \psi, \psi(t/x)\hor_w$}
\RightLabel{$\existsr$}
\UnaryInfC{$\ns \hol t, \ex x \psi \hor_w$}
\DisplayProof
\end{tabular}
\end{center}


\begin{center}
\begin{tabular}{c c c c}
\AxiomC{$\ns \hol y, \phi(y/x) \hor$}
\RightLabel{$\allr^{\dag_{1}}$}
\UnaryInfC{$\ns \hol \fa x\phi \hor$}
\DisplayProof

&

\AxiomC{$\ns \hol \Diamond\phi \hor_{w} \hol \phi \hor_{u}$}
\RightLabel{$\diar^{\dag_{2}(\fcons)}$}
\UnaryInfC{$\ns \hol \Diamond\phi \hor_{w} \hol \emptyset \hor_{u}$}
\DisplayProof

&

\AxiomC{$\ns \hol [ \phi] \hor $}
\RightLabel{$\boxr$}
\UnaryInfC{$\ns \hol  \Box \phi \hor$}
\DisplayProof

&

\AxiomC{$\ns \hol t \neq t \hor$}
\RightLabel{$\nidref$}
\UnaryInfC{$\ns \hol  \emptyset \hor$}
\DisplayProof
\end{tabular}
\end{center}


\begin{center}
\begin{tabular}{c c c}
\AxiomC{$\ns \hol t \neq s, N(t/z), N(s/z) \hor$}
\RightLabel{$\idrep$}
\UnaryInfC{$\ns \hol t \neq s, N(t/z)  \hor$}
\DisplayProof

&

\AxiomC{$\ns \hol s,t ,t \neq s \hor$}
\RightLabel{$\drep$}
\UnaryInfC{$\ns \hol t, t \neq s \hor$}
\DisplayProof

&

\AxiomC{$\ns \hol s\neq t \hor_w \hol s\neq t \hor_u$}
\RightLabel{$\idrig^{\dag_4}$}
\UnaryInfC{$\ns \hol  s\neq t \hor_w \hol \emptyset \hor_u$}
\DisplayProof

\end{tabular}
\end{center}


\begin{center}
\begin{tabular}{c c c c}

\AxiomC{$\ns \hol t \hor_{w} \hol t \hor_{u}$}
\RightLabel{$\dpr^{\dag_{3}(\fcons)}$}
\UnaryInfC{$\ns \hol  t \hor_{w} \hol \emptyset \hor_{u}$}
\DisplayProof
&

\AxiomC{$\ns \hol  [ \emptyset] \hor$}
\RightLabel{$\drule$}
\UnaryInfC{$\ns \hol \emptyset \hor$}
\DisplayProof

&

\AxiomC{$\ns \hol y \hor$}
\RightLabel{$\ned^{\dag_{1}}$}
\UnaryInfC{$\ns \hol \emptyset \hor$}
\DisplayProof

&

\AxiomC{$\ns \hol  t \hor$}
\RightLabel{$\cdr$}
\UnaryInfC{$\ns \hol \emptyset \hor$}
\DisplayProof
\end{tabular}
\end{center}


\begin{flushleft}
\textbf{Side Conditions:} Let $\gpset = \{\gpc(n,k) \ | \ \gpc(n,k) \in \fcons\}$.\\
\begin{minipage}{0.4\textwidth}$\dag_{1} :=$ $y$ is fresh.\\
$\dag_{2}(\fcons):=$ $w \prpath{L} u$ with $L = L_{\thuesys(\gpset)}(\fd)$.
\end{minipage}
\begin{minipage}{0.4\textwidth}
$\dag_{3}(\fcons) :=$  $w\neq u$ and: see \fig~\ref{fig:side-conditions}.\\
$\dag_{4}:=$  $w\neq u$.
\end{minipage}
\end{flushleft}

\caption{The nested system $\nql$ with $\fcons$ a set of frame conditions. 
\label{fig:nested-calculi}}
\begin{center}
\bgroup
\setlength{\tabcolsep}{5pt}
\def\arraystretch{1.1}
\begin{tabular}{| l | l |}
\hline
Frame Conditions in $\fcons$ &  Corresponding Additional Side Conditions for $\dpr$\\
\hline
$\cid \in \fcons$, but $\cdd \not\in \fcons$ & $\dag_{3}(\fcons) :=$ ``$w \prpath{L} u$ with $L := L_{\sfour \cup \thuesys(\gpset)}(\fd)$.''\\
$\cdd \in \fcons$, but $\cid \not\in \fcons$ & $\dag_{3}(\fcons) :=$ ``$w \prpath{L} u$ with $L := L_{\sfour \cup \thuesys(\gpset)}(\bd)$.''\\
$\cid, \cdd \in \fcons$ & $\dag_{3}(\fcons) :=$ ``$w \prpath{L} u$ with $L := L_{\sfive(\gpset)}(\fd)$.''\\
\hline
\end{tabular}
\egroup
\end{center}

\caption{The side condition $\dag_{3}(\fcons)$ for $\dpr$.}
\label{fig:side-conditions}
\end{figure}

Depending on the contents of $\fcons$, a nested calculus $\nql$ may contain (none/some/all of) the rules $\set{\dpr, \ned, \cdr, \drule}$ (see \dfn~\ref{def:nested-systems} for details). 
 The \emph{domain propagation rule} $\dpr$ semantically captures the fact that if a term is interpreted in the inner domain of a world, then it is interpreted in the inner domains of all future worlds, all past worlds, or all worlds, depending on if the model has increasing, decreasing, or constant inner domains, respectively. The \emph{classical domains rule} $\cdr$ semantically captures classical domains by letting one bottom-up introduce any term to a component of a nested sequent; this corresponds to the fact that all terms are interpreted in all domains of a model satisfying classical domains. Note that the $\allr$ and $\ned$ rules are subject to a freshness condition, that is, the variable $y$ must be \emph{fresh} in any application of the rule, i.e., $y$ cannot occur free in the conclusion of a rule application. In the \emph{non-empty domains rule} $\ned$, this freshness condition encodes the fact that \emph{some} object exists in the inner domain of each world. Last, the \emph{seriality rule} $\drule$ corresponds to the seriality condition. The exact semantic reading of each inference rule is made clear in the proof of soundness (see \thm~\ref{thm:soundness-nested}).

\begin{remark}\label{rmk:all-rule-subsumes-EBR} Observe that when the premise $\ns$ of an $\allr$ instance has a formula interpretation of the form $\fm(\ns) = \phi_{0} \lor \Box (\phi_{1} \lor \Box (\cdots \lor \Box (\phi_{n} \lor \Box \phi_{n+1})\ldots))$, then the $\allr$ rule application is essentially an application of $\ebr$. This shows that $\ebr$ is essentially a special instance of $\allr$, i.e., $\allr$ is stronger than and subsumes the $\ebr$ rule. As mentioned in Remark~\ref{rmk:constant-outer}, the inclusion of $\ebr$ in a QML forces the outer domains of normal models to be constant (see Corsi~\cite[pp. 1503-1504]{Cor02}), which explains why our nested sequent systems are sound and complete relative to QMLs with constant outer domains.
\end{remark}

As mentioned above, the parameter $\fcons$ dictates what logic $\ql$ the system $\nql$ is sound and complete relative to. This is achieved by (1) adding the $\dpr$, $\drule$, $\ned$, or $\cdr$ rules to the system and/or (2) changing the functionality of the reachability rules $\diar$ and  $\dpr$. Reachability rules are special in that they view nested sequents as automata, and enable both terms and formulae to be (bottom-up) propagated or consumed along certain paths (corresponding to strings generated by a $\albet$-system) within the tree $\tr(\ns)$ of a nested sequent $\ns$. To make the functionality of such rules precise, we define \emph{propagation graphs} and \emph{propagation paths} (cf.~\cite{CiaLyoRamTiu21,GorPosTiu11}). 

\begin{definition}[Propagation Graph]\label{def:propagation-graph} Let $\ns$ be a nested sequent with $\tr(\ns) = (V,E)$. We define the \emph{propagation graph} $\prgr{\ns} = (\prv,\pre,\prl)$ such that 

\begin{itemize}

\item[(1)] $w \in \prv$ \iffi $w$ is the name of a component that exists in $\ns$;

\item[(2)] $(w, \fd, u), (u, \bd, w) \in \pre$ \iffi $(w,u) \in E$;

\item[(3)] $\prl(w) = \vt$ \iffi $(w, \vt, \msi) \in V$.

\end{itemize}
We will often write $w \in \prgr{\ns}$ to mean $w \in \prv$, and $(w,\ques,u) \in \prgr{\ns}$ to mean $(w,\ques,u) \in \pre$.
\end{definition}

\begin{definition}[Propagation Path]\label{def:propagation-path} Given a propagation graph $\prgr{\ns} = (\prv,\pre,\prl)$, two names $u,w \in \prv$, and a character $\ques \in \albet$, we write $\prgr{\ns} \models u \prpath{\ques} w$ \iffi $(u,\ques,w) \in \pre$. Moreover, given a string $\ques\stra \in \albetstr$ where $\ques \in \albet$, we inductively define $\prgr{\ns} \models u \prpath{\ques\stra} w$ as `$\exists_{v \in \prv} \ \prgr{\ns}  \models u\prpath{\ques} v$ and $\prgr{\ns} \models v \prpath{\stra}w$', and we take $\prgr{\ns} \models u \prpath{\empstr} w$ to mean that $u = w$. Additionally, when $\prgr{\ns}$ is clear from the context we may simply write $u \prpath{\stra} w$ to express $\prgr{\ns} \models u \prpath{\stra} w$. Finally, given a language $\glang(\stra)$ of some $\albet$-system $\g$ and $\stra \in \albetstr$, we write $u \ \prpath{\glang(\stra)} \ w$ \iffi there is a string $\strb \in \glang(\stra)$ such that $u \prpath{\strb} w$.
\end{definition}

The following lemma is a consequence of \lem~\ref{lem:S4-with-S(G)-is-S4-or-S5} and will be useful in the sequel.

\begin{lemma}\label{lem:reversing-paths}
Let $\ns$ be a nested sequent, $\gpset$ be a set of generalized path conditions, $L = L_{\g(\gpset)}(\ques)$, and $L^{-1} = L_{\g(\gpset)}(\conv{\ques})$. Then, $\prgr{\ns} \models u \prpath{L} w$ \iffi $\prgr{\ns} \models w \ \prpath{L^{-1}} \ u$.
\end{lemma}

With the above notions, one can formally specify the operation of the reachability rules $\diar$ and $\dpr$. The side conditions dictating how such rules can be applied are listed in Figures~\ref{fig:nested-calculi} and~\ref{fig:side-conditions}. 
To make the operation of reachability rules clearer for the reader, we provide an example of a $\diar$ 
rule application below.

\begin{example}\label{ex:propagation-graph-path} Let $\ns = x, y, \psi, [\exists x P(x), P(y)], [a, \dia (\phi \lor \psi), \phi \lor \psi]$. Reading $\ns$ from left to right, we assume the names of the three components are $w$, $v$, and $u$. A graphical depiction of the propagation graph $\prgr{\ns} = (\prv,\pre,\prl)$ is given below, where we denote each vertex $w' \in \prv = \{w,v,u\}$ as a pair $(w',\prl(w'))$ to make its label clear.
\begin{center}
\begin{minipage}[t]{.5\textwidth}
\xymatrix{
  (v,\emptyset)\ar@/^-1pc/@{.>}[rr]|-{\bd} & &  (w,\{x,y\})\ar@/^1pc/@{.>}[rr]|-{\fd}\ar@/^-1pc/@{.>}[ll]|-{\fd} & &  (u,\{a\})\ar@/^1pc/@{.>}[ll]|-{\bd}
}
\end{minipage}
\end{center}
If $\gpc(1,1) \in \fcons$ (i.e., Euclideanity is a frame condition), then $\fd \pto \bd \cate \fd \in \thuesys(\gpset)$, and since $\prgr{\ns} \models u \prpath{\bd \fd} u$ with $\bd \cate \fd \in L_{\thuesys(\gpset)}(\fd)$, it follows that the $\diar$ rule can be applied to $\ns$, letting us derive the nested sequent $x, y, \psi, [\exists x P(x), P(y)], [a, \dia (\phi \lor \psi)]$. 
\end{example}

\begin{definition}[Nested Sequent Systems]\label{def:nested-systems} We define the base nested sequent system to be the following collection of rules:
$$
\nqk := \set{\ax, \disr, \conr, \existsr, \allr, \diar, \boxr, \nidref, \idrep, \idrig, \drep}.
$$
We take $\nqk = \nqk(\emptyset)$, meaning, the side condition for the 
$\diar$ rule is $\dag_{2}(\emptyset) =$ ``$w \prpath{R} u$''.  For each set $\fcons$ of frame conditions, we define the nested sequent system $\nql$ to be an extension of $\nqk$ with rules from the set $\set{\dpr, \ned, \cdr, \drule}$ that satisfies the following conditions:
\begin{itemize}

\item $\dpr \in \nql$ \iffi $\fcons \cap \set{\cid, \cdd} \neq \emptyset$;

\item $\ned \in \nql$ \iffi $\cned \in \fcons$;

\item $\cdr \in \nql$ \iffi $\ccd \in \fcons$;

\item $\drule \in \nql$ \iffi $\cser \in \fcons$.

\end{itemize}
Each reachability rule from the set $\set{\diar, \dpr}$ occurring in $\nql$ is assumed to satisfy its side condition as stipulated in \fig~\ref{fig:nested-calculi} and \fig~\ref{fig:side-conditions}.
\end{definition}

As usual, we refer to terms and formulae that are explicitly presented in the premises and conclusion of a rule as \emph{auxiliary} and \emph{principal}, respectively. For example, in the $\allr$ rule, $t$ and $\phi(t/x)$ are auxiliary and $\forall x\phi$ is principal. A \emph{derivation} in $\nql$ of a nested sequent $\ns$ is a (potentially infinite) tree whose nodes are labeled with nested sequents such that: 
\begin{itemize}

\item[(1)] The root is labeled with $\ns$;

\item[(2)] Every parent node is the conclusion of an instance of a  rule of $\nql$ with its children the premises.


\end{itemize}
Without loss of generality, we identify derivations that differ only in the names of bound variables,\label{note:nested-equiv-renaming} we assume that no variable has both a free and bound occurrence in a derivation, and we assume that every fresh variable used in a proof is \emph{globally fresh}, meaning there is a one-to-one correspondence between $\allr$ (and $\ned$) applications and their fresh variables. These assumptions are justified by our identification of formulae modulo the names of bound variables and by 
Lemma~\ref{lem:ps-admiss}. We define a \emph{branch} $\branch = \ns_{0}, \ns_{1}, \ldots, \ns_{n}, \ldots$ to be a maximal path of nested sequents in a derivation such that $\ns_{0}$ is the conclusion of the derivation and each nested sequent $\ns_{i+1}$ (if it exists) is a child of $\ns_{i}$. 

A \emph{proof} is a finite derivation where all leaves are instances of $\ax$. We use $\prf$ and annotated versions thereof to denote both derivations and proofs with the context differentiating the usage. If a proof of a nested sequent $\ns$ exists in $\nql$, then we write $\nql\vdash\ns$ to indicate this. Often, when discussing a proof transformation, we will use the following notation to denote that $\prf$ is a proof of a nested sequent $\ns$: 
\begin{center}
\AxiomC{$\dotprf{\prf}$}
\UnaryInfC{$\ns$}
\DisplayProof
\end{center}
The \emph{height} of a derivation is defined in the usual way as the length of a maximal branch in the derivation, which may be infinite if the derivation is infinite.

\begin{example} We provide examples of two proofs below. The proof shown below left proves an instance of the Barcan Formula $\abf$ 
and we suppose that the proof is given in a nested system $\nql$ with $\cdd \in \fcons$. The proof shown below right proves an instance of the Converse Barcan Formula $\acbf$ 
and we suppose that the proof is given in a nested system $\nql$ with $\cid \in \fcons$. Observe that the application of $\dpr$ is allowed in the left proof as $\cdd \in \fcons$ and the application of $\dpr$ is allowed in the right proof as $\cid \in \fcons$. In particular, the $\dpr$ rule is needed in the left proof to (bottom-up) shift the variable $y$ to the root of the nested sequent and the $\dpr$ rule is needed in the right proof to (bottom-up) shift the variable $y$ into the nesting $[P(y)]$; without this functionality neither proof can be (bottom-up) completed.\\
\resizebox{\textwidth}{!}{
\begin{tabular}{c c}
\AxiomC{$\phantom{P}$}
\RightLabel{$\ax$}
\UnaryInfC{$y, \exists x \dia \ngn{P(x)}, \dia  P(y) [y, P(y), \ngn{P(x)}]$}
\RightLabel{$\diar$}
\UnaryInfC{$y, \exists x \dia \ngn{P(x)}, \dia \ngn{P(y)} [y, P(y)]$}
\RightLabel{$\existsr$}
\UnaryInfC{$y, \exists x \dia \ngn{P(x)}, [y, P(y)]$}
\RightLabel{$\dpr$}
\UnaryInfC{$\exists x \dia \ngn{P(x)}, [y, P(y)]$}
\RightLabel{$\allr$}
\UnaryInfC{$\exists x \dia \ngn{P(x)}, [\forall x P(x)]$}
\RightLabel{$\boxr$}
\UnaryInfC{$\exists x \dia \ngn{P(x)}, \Box \forall x P(x)$}
\RightLabel{$\disr$}
\UnaryInfC{$\exists x \dia \ngn{P(x)} \lor \Box \forall x P(x)$}
\DisplayProof

&

\AxiomC{$\phantom{P}$}
\RightLabel{$\ax$}
\UnaryInfC{$y, \dia \exists x \ngn{P(x)}, [y, P(y), \exists x \ngn{P(x)}, \ngn{P(y)}]$}
\RightLabel{$\existsr$}
\UnaryInfC{$y, \dia \exists x \ngn{P(x)}, [y, P(y), \exists x \ngn{P(x)}]$}
\RightLabel{$\diar$}
\UnaryInfC{$y, \dia \exists x \ngn{P(x)}, [y, P(y)]$}
\RightLabel{$\dpr$}
\UnaryInfC{$y, \dia \exists x \ngn{P(x)}, [P(y)]$}
\RightLabel{$\boxr$}
\UnaryInfC{$y, \dia \exists x \ngn{P(x)}, \Box P(y)$}
\RightLabel{$\allr$}
\UnaryInfC{$\dia \exists x \ngn{P(x)}, \forall x \Box P(x)$}
\RightLabel{$\disr$}
\UnaryInfC{$\dia \exists x \ngn{P(x)} \lor \forall x \Box P(x)$}
\DisplayProof
\end{tabular}
}
\end{example}

\subsection{Soundness and Completeness}

We first prove the soundness of our nested systems, and afterward, discuss the proof of completeness. Let us first define a few concepts that are of use in proving soundness. 

Let $\M = \<\W,\R,\U, \D,\I\>$ be a model. We define $\R^{-1}$ to be the inverse of $\R$, i.e., for any $w, u \in \W$, $u \R^{-1} w$ \iffi $w \R u$. We define $\rho(\fd) := \R$ and $\rho(\bd) := \R^{-1}$. For $w,u \in \W$, we define $\M \models w \prpath{\ques} u$ \iffi $w \rho(\ques) u$. 
 Given a string $\ques\stra \in \albetstr$ where $\ques \in \albet$, we inductively define $\M \models w \prpath{\ques\stra} u$ as `$\exists_{v \in \W} \ \M \models w \prpath{\ques} v$ and $\M \models v \prpath{\stra} u$', and we take $\M \models w \prpath{\empstr} u$ to mean that $w = u$. Given a language $\glang(\stra)$ of some $\albet$-system $\g$ with $\stra \in \albetstr$, we write $\M \models w \ \prpath{\glang(\stra)} \ u$ \iffi  there is a string $\strb \in \glang(\stra)$ such that $\M \models w \prpath{\strb} u$.

\begin{lemma}\label{lem:paths-implies-edge-in-model}
Let $\M = \<\W,\R,\U, \D,\I\>$ be a model based on a frame in $\fclassc$, $w, u \in \W$, and $\gpset = \set{\gpc(n,k) \mid \gpc(n,k) \in \fcons \text{ and } n,k \in \mathbb{N}}$. If $\M \models w \prpath{L} u$ with $L = L_{\thuesys(\gpset)}(\ques)$, then $w \rho(\ques) u$.
\end{lemma}

\begin{proof} By induction on the length of the derivation of $\stra \in L = L_{\thuesys(\gpset)}(\ques)$ such that $\M \models w \prpath{\stra} u$. 

\textit{Base case.} For the base case, suppose the length of the derivation of $\stra$ is $0$. By \dfn~\ref{def:semi-thue-deriv-lang}, the only derivation in $L_{\thuesys(\gpset)}(\fd)$ of length $0$ is the derivation consisting solely of $\ques$. Hence, $\stra = \ques$, meaning $\M \models w \prpath{\ques} u$, so $w \rho(\ques) u$ holds by definition.

\textit{Inductive step.} Let the derivation of $\stra$ be of length $n+1$. By \dfn~\ref{def:semi-thue-deriv-lang}, there is a derivation $\fd \pto^{*}_{\thuesys(\gpset)} \strb \ques_{0} \strc$ of length $n$ and production rule $\ques_{0} \pto \ques_{1} \cdots \ques_{m} \in \thuesys(\gpset)$ such that $\stra = \strb \ques_{1} \cdots \ques_{m} \strc$. Since $\M \models w \prpath{\stra} u$, it follows that there exist worlds $v_{0}, \ldots, v_{m} \in \W$ such that $v_{0} \rho(\ques_{1}) v_{1}, \ldots v_{m-1} \rho(\ques_{m}) v_{m}$. By definition, there is a generalized path condition corresponding to the production rule $\ques_{0} \pto \ques_{1} \cdots \ques_{m} \in \thuesys(\gpset)$, which implies that $v_{0} \rho(\ques_{0}) v_{m}$. Therefore, $\M \models w \prpath{\strb} v_{0}$, $\M \models v_{0} \prpath{\ques_{0}} v_{m}$, and $\M \models v_{m} \prpath{\strb} u$, meaning $\M \models w \ \prpath{\strb \ques_{0} \strc} \ u$. Since the string $\strb \ques_{0} \strc$ has a derivation of length $n$, it follows that $w \rho(\ques) u$ by IH.
\end{proof}

\begin{lemma}\label{lem:paths-imply-relation}
Let $\M = \<\W,\R,\U, \D,\I\>$ be a model based on a frame in $\fclassc$, $w, u \in \W$, $\s$ be an $\M$-assignment, and $\mint$ be an $\M$-interpretation. Moreover, let $\gpset = \set{\gpc(n,k) \mid \gpc(n,k) \in \fcons \text{ and } n,k \in \mathbb{N}}$. If $\M, \s, \mint \not\Vdash \ns$ and $\prgr{\ns} \models w \prpath{L} u$ with $L = L_{\thuesys(\gpset)}(\ques)$, then $\mint(w) \rho(\ques) \mint(u)$.
\end{lemma}

\begin{proof} By assumption, it follows that there exists a string $\stra \in L$ such that $\prgr{\ns} \models w \prpath{\stra} u$. By \dfn~\ref{def:sequent-semantics}, it follows that $\M \models \mint(w) \prpath{\stra} \mint(u)$, and so, $\mint(w) \rho(\ques) \mint(u)$ by \lem~\ref{lem:paths-implies-edge-in-model}.
\end{proof}

\begin{theorem}[Soundness]\label{thm:soundness-nested} If $\nql \vdash \ns$, then $\ns$ is valid w.r.t. the class of frames $\fclassc$.
\end{theorem}

\begin{proof} Let us fix a class $\fclassc$ of frames. When we speak of (in)validity below, we mean that a nested sequent is (in)valid relative to this class of frames. It is straightforward to show that every instance of $\ax$ is valid. Therefore, we argue soundness by showing that the other rules of $\nql$ are sound, i.e., if the conclusion is invalid, then at least one premise is invalid. We present the $\existsr$, $\diar$ 
and $\dpr$ cases; additional cases are proven in Appendix~\ref{app:proofs-sec-sound-compl}.

$\existsr$. Suppose that $\nsii = \ns \hol t, \Gamma, \ex x \psi \hor_w$ is invalid. Then, there exists a model $\M$ based on a frame in $\fclassc$, an $\M$-assignment $\s$, and an $\M$-interpretation $\mint$ such that $\M, \s, \mint \not\Vdash \nsii$. By \dfn~\ref{def:sequent-semantics}, $\I_{\mint(w)}^{\s}(t) \in \D_{\mint(w)}$ and $\M, \s, \mint(w) \not\Vdash \exists x \psi$, meaning, $\M, \s, \mint(w) \not\Vdash \psi(t/x)$. This shows that $\M, \s, \mint \not\Vdash  \ns \hol t, \Gamma, \ex x \psi, \psi(t/x)\hor_w$.


$\diar$. Suppose that $\nsii = \ns \hol  \Gamma, \Diamond\phi \hor_{w} \hol \Delta \hor_{u}$ is invalid. By assumption, there exists a model $\M$ based on a frame in $\fclassc$, an $\M$-assignment $\s$, and an $\M$-interpretation $\mint$ such that $\M, \s, \mint \not\Vdash \nsii$, which further implies that $\M, \s, \mint(w) \not\Vdash \Diamond \phi$. By the side condition on $\diar$, we know that $w \prpath{L} u$ with $L = L_{\thuesys(\gpset)}(\fd)$. By \lem~\ref{lem:paths-imply-relation} it follows that $\mint(w)\R\mint(u)$. Hence, $\M, \s, \mint(u) \not\Vdash \phi$. This shows that $\M, \s, \mint \not\Vdash \ns \hol \Gamma, \Diamond\phi \hor_{w} \hol \Delta, \phi \hor_{u}$.


$\dpr$. Suppose that $\nsii = \ns \hol  t, \Gamma \hor_{w} \hol \Delta \hor_{u}$ is invalid. We show the case where $\cid \in \fcons$ and $\cdd \not\in \fcons$. By assumption, there exists a model $\M$ based on a frame in $\fclassc$, an $\M$-assignment $\s$, and an $\M$-interpretation $\mint$ such that $\M, \s, \mint \not\Vdash \nsii$, which means that $\s(t) \in \D_{\mint(w)}$. By the side condition on $\drep$, we know that $w \prpath{L} u$ with $L := L_{\sfour \cup \thuesys(\gpset)}(\fd)$. By \lem~\ref{lem:S4-with-S(G)-is-S4-or-S5}, we know that either $L = L_{\sfour}(\fd)$ or $L = L_{\sfive}(\fd)$. We prove the first case as the second case is similar. Since $w \ \prpath{L} \ u$ with $L = L_{\sfour}(\fd)$, we know that there exists a sequence of names $v_{1}, \ldots, v_{n}$ such that $w \prpath{\fd} v_{1}, v_{1} \prpath{\fd} v_{2}, \ldots, v_{n} \prpath{\fd} u$. By Definitions~\ref{def:sequent-semantics} and~\ref{def:propagation-path}, $\mint(w)\R\mint(v_{1}), \mint(v_{1})\R\mint(v_{2}), \ldots, \mint(v_{n})\R\mint(u)$. By the $\cid$ condition and the fact that $\s(t) \in \D_{\mint(w)}$, we know that $\s(t) \in \D_{\mint(u)}$. Hence, $\M, \s, \mint \not\Vdash \ns \hol t, \Gamma \hor_{w} \hol t, \Delta \hor_{u}$. 
\end{proof}

Completeness is proven by extracting a counter-model from failed (potentially non-terminating) proof-search. To be specific, we assume we are given an unprovable nested sequent $\ns$ and (bottom-up) apply rules from $\nql$ in an exhaustive and fair manner, which yields a (potentially infinite) derivation of $\ns$. As $\ns$ is unprovable in $\nql$, this derivation cannot be a proof, meaning, there must be a branch in the proof that does not terminate at an instance of $\ax$. Using this branch, a model can be extracted that falsifies $\ns$, which proves the contrapositive of the completeness statement below (\thm~\ref{thm:completeness-nested}). 

Although this completeness strategy is common in the domain of proof theory, the actual implementation of this strategy in our setting proves to be rather complex. This is due to the large number of details one must keep track of and the non-triviality of extracting a counter-model from an infinite branch of nested sequents. Due to the length of the completeness proof, we have deferred it to the appendix; the reader can consult Appendix~\ref{app:proofs-sec-sound-compl} for the full details.

\begin{theorem}[Completeness]\label{thm:completeness-nested}  If $\ns$ is valid w.r.t. the class of frames $\fclassc$, then $\nql \vdash \ns$.
\end{theorem}

%% file: body-proof-prop2.tex
In this section, we show that each nested calculus $\ql$ satisfies a wide array of interesting properties, which we leverage in our proof of syntactic cut-elimination. In particular, we show that various rules are \emph{(height-preserving) admissible}, meaning if the premises of the rule have proofs (of height $h_{1}, \ldots, h_{n}$), then the conclusion of the rule has a proof (of height $h \leq \max\{h_{1}, \ldots, h_{n}\}$). We refer to a height-preserving admissible rule as \emph{hp-admissible}.

\begin{figure}
\centering

\begin{tikzpicture}
\node[] (a) [] {$\ps$};
\node[] (b) [right of= a, xshift=1cm] {$\tw$};
\node[] (d) [below of= a] {$\ew$};
\node[] (c) [right of= d, xshift=1cm] {$\wk$};
\node[] (f) [right of= c, xshift=1cm] {Lem.~\ref{lem:hp-inver}};
\node[] (g) [below of= c] {$\tc$};
\node[] (e) [below of= d] {$\nec$};
\node[] (h) [right of= g, xshift=1cm] {$\ctr$};
\node[] (i) [below of= h, xshift=1cm] {$\ec$};
\node[] (j) [below of= e] {$\shift$};
\node[] (k) [right of= f, xshift=1cm] {$\ebr$};

\path[->] (b) edge[bend left=25] node[above] {} (k);
\path[->] (a) edge[] node[above] {} (b);
\path[->] (a) edge[] node[above] {} (c);

\path[->] (b) edge[] node[above] {} (f);
\path[->] (c) edge[] node[above] {} (f);
\path[->] (d) edge[bend right=25] node[above] {} (f);

\path[->] (g) edge[] node[above] {} (h);
\path[->] (g) edge[bend right=25] node[above] {} (i);
\path[->] (h) edge[bend right=25] node[above] {} (i);
\path[->] (i) edge[bend right=25] node[above] {} (h);

\path[->] (f) edge[bend right=25] node[above] {} (h);
\path[->] (f) edge[bend left=25] node[above] {} (i);

\path[->] (f) edge[] node[above] {} (k);
\end{tikzpicture}

\caption{A diagram depicting which (hp-)admissibility (and hp-invertibility) results are sufficient to prove others. An arrow from one rule to another indicates that the (hp-)admissibility of the source rule is sufficient to prove the (hp-)admissibility of the target rule. Note that `Lem.~\ref{lem:hp-inver}' indicates that all rules in $\nql$ are hp-invertible, i.e., that the $i$-inverse of all rules in $\nql$ are hp-admissible.\label{fig:admiss-invert-dependencies}}
\end{figure}
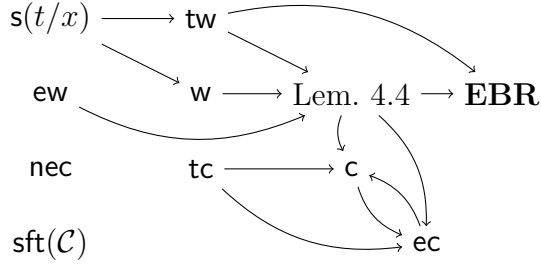

We will show all rules in \fig~\ref{fig:admiss-rules} hp-admissible in $\nql$, i.e., we prove the \emph{substitution rule} $\ps$, \emph{weakening rule} $\wk$, \emph{term weakening rule} $\tw$, \emph{external weakening rule} $\ew$, \emph{necessitation rule} $\nec$, \emph{contraction rule} $\ctr$, \emph{term contraction rule} $\tc$, \emph{external contraction rule} $\ec$, and \emph{shift rule} $\shift$ hp-admissible. These hp-admissibility results are crucial for establishing cut-elimination in the subsequent section. Regarding the substitution rule $\ps$, we define $\ns(t/x)$ to be the nested sequent $\ns$ obtained by replacing every occurrence of the variable $x$ in $\ns$ by $t$, regardless of if $x$ occurs in a signature or formula. Notice also that substitution has been defined in such a way that it may involve renaming of bound variables to avoid capture of free ones and, hence, it is not need to require $t$ to be free for $x$ in the formulas where we apply $\ps$. As we will see below, various hp-admissibility results will be used to prove further structural rules hp-admissible. To make the dependencies between these various lemmas clear, we have provided a diagram in \fig~\ref{fig:admiss-invert-dependencies} that displays which (hp-)admissible rules are sufficient to prove other rules (hp-)admissible.

The shift rule $\shift$ is a novel contribution of this paper and is interesting as it uniformly captures reasoning with generalized path conditions in a single rule.\footnote{Note that $\ncon, \nsii := \ncon,\nant,[\nsii_{1}], \ldots,[\nsii_{n}]$ for a flat sequent $\ncon$ and nested sequent $\nsii = \nant, [\nsii_{1}], \ldots, [\nsii_{n}]$.} This is advantageous as it does not require one to introduce a special structural rule for \emph{each} individual generalized path condition, as is the approach in other settings~\cite{Bru09,LyoOrl23} and which requires separate (hp-)admissibility proofs for each added rule. Moreover, as pointed out by Brünnler~\cite{Bru09}, having a distinct structural rule for each generalized path condition can lead to a loss of modularity, requiring the addition of even more structural rules to secure completeness. This allows for us to prove the hp-admissibility of a single, generic rule that simplifies our proof of cut-elimination and makes the proof uniform over the class of QMLs we consider. 

Beyond establishing the hp-admissibility of the rules in \fig~\ref{fig:admiss-rules}, we prove that all non-initial rules of $\nql$ are \emph{height-preserving invertible}. If we let $\ru^{-1}_{i}$ be the $i$-inverse of the rule $\ru$ whose conclusion is the $i^{th}$ premise of the $n$-ary rule $\ru$ and premise is the conclusion of $\ru$, then we say that $\ru$ is \emph{(height-preserving) invertible} \iffi $\ru^{-1}_{i}$ is (height-preserving) admissible for each $1 \leq i \leq n$. We refer to height-preserving invertible rules as \emph{hp-invertible}.

In the following subsection, we prove a sequence of hp-admissibility and hp-invertibility results, which will be sufficient to establish cut-elimination in the \sect~\ref{subsec:cut-elim}.

\begin{figure}[t]

\begin{center}
\begin{tabular}{c c c c c}
\AxiomC{$\ns$}
\RightLabel{$\ps$}
\UnaryInfC{$\ns(t/x)$}
\DisplayProof

&

\AxiomC{$\ns \hol \emptyset \hor$}
\RightLabel{$\wk$}
\UnaryInfC{$\ns \hol   \phi \hor$}
\DisplayProof

&
\AxiomC{$\ns \hol \emptyset \hor$}
\RightLabel{$\tw$}
\UnaryInfC{$\ns \hol t \hor$}
\DisplayProof
&

\AxiomC{$\ns \hol \emptyset \hor$}
\RightLabel{$\ew$}
\UnaryInfC{$\ns \hol  [\emptyset] \hor$}
\DisplayProof

&

\AxiomC{$\ns$}
\RightLabel{$\nec$}
\UnaryInfC{$[ \ns ]$}
\DisplayProof
\end{tabular}
\end{center}


\begin{center}
\begin{tabular}{c c c c}
\AxiomC{$\ns \hol  \phi, \phi  \hor$}
\RightLabel{$\ctr$}
\UnaryInfC{$\ns \hol  \phi  \hor$}
\DisplayProof

&

\AxiomC{$\ns \hol t,t   \hor$}
\RightLabel{$\tc$}
\UnaryInfC{$\ns \hol t \hor$}
\DisplayProof

&

\AxiomC{$\ns \hol [ \nant ], [ \ncon ] \hor$}
\RightLabel{$\ec$}
\UnaryInfC{$\ns \hol [ \nant, \ncon ] \hor$}
\DisplayProof

&

\AxiomC{$\ns \hol  \nant , [\nsii] \hor_{w} \hol  \ncon  \hor_{v}$}
\RightLabel{$\shift^{\dag(\fcons)}$}
\UnaryInfC{$\ns \hol  \nant  \hor_{w} \hol  \ncon, \nsii \hor_{v}$}
\DisplayProof
\end{tabular}
\end{center}


\begin{flushleft}
\textbf{Side Conditions:}\\
$\dag(\fcons) := w \prpath{L} v$ with $L = L_{\g(\gpset)}(\fd)$.
\end{flushleft}

\caption{Height-preserving admissible rules in $\nql$.\label{fig:admiss-rules}}
\end{figure}

\subsection{Admissibility and Invertibility Properties}

We first show that a `generalized form' of the $\ax$ rule is admissible in $\nql$. That is to say, while the $\ax$ rule stipulates that 
all nested sequents of the form $\ns \hol  L, \negnnf{L}  \hor$ with $L$ and $\negnnf{L}$ \emph{literals} are provable in $\nql$, the following lemma determines that all nested sequents of the form $\ns \hol  \phi, \negnnf{\phi}  \hor$ with $\phi, \negnnf{\phi} \in \lang$ are provable in $\nql$. This property is crucial to ensure that proofs are closed under the substitution of arbitrary formulae for propositional atoms, which is an essential factor in determining completeness (e.g., see the proof of~\cite[Lemma 12]{Lyo21thesis}).

\begin{lemma}[Generalized Axioms]\label{lem:generalized-id}
For any $\phi  \in \lang$, $\ns \hol  \phi, \negnnf{\phi}  \hor$ is provable in $\nql$.
\end{lemma}

\begin{proof} The result is shown by induction on the length of $\phi$.\footnote{We use `IH' to indicate an application of the induction hypothesis.} We show the case where $\phi$ is of the form $\forall x \psi$ or $\Box \phi$ as the remaining cases are simple or similar.
\begin{center}
\begin{tabular}{c @{\hskip 1cm} c}
\AxiomC{$\dotprf{\prf}$}
\RightLabel{IH}
\UnaryInfC{$\ns \hol y,\exists x \negnnf{\psi} ,\psi(y/x),\negnnf{\psi(y/x)} \hor_{w} $}
\RightLabel{$\existsr$}
\UnaryInfC{$\ns \hol  y, \exists x \negnnf{\psi} ,\psi(y/x) \hor_{w} $}
\RightLabel{$\allr$}
\UnaryInfC{$\ns \hol  \exists x \negnnf{\psi} , \forall x \psi \hor_{w} $}
\DisplayProof

&

\AxiomC{$\dotprf{\prf}$}
\RightLabel{IH}
\UnaryInfC{$\ns \hol \Diamond \negnnf{\psi} , [ \psi,\negnnf{\psi}] \hor_{w} $}
\RightLabel{$\diar$}
\UnaryInfC{$\ns \hol \Diamond \negnnf{\psi} , [ \psi] \hor_{w} $}
\RightLabel{$\boxr$}
\UnaryInfC{$\ns \hol  \Diamond\negnnf{ \psi} , \Box \psi \hor_{w} $}
\DisplayProof
\end{tabular}
\end{center}
Note that the application of the $\diar$ rule in the above right proof is \emph{always} permitted in $\nql$, regardless of the contents of $\fcons$. This is because $\fd \in L_{\thuesys(\gpset)}(\fd)$ by \dfn~\ref{def:semi-thue-deriv-lang}.
\end{proof}

The following lemmas (\ref{lem:ps-admiss}--\ref{lem:shift-admiss}) are all shown by induction on the height of the given proof of the premise.

\begin{lemma}[Substitution]\label{lem:ps-admiss}
The $\ps$ rule is hp-admissible in $\nql$.
\end{lemma}

\begin{proof} The base case is trivial as any application of $\ps$ to $\ax$  yields another instance of the rule. With the exception of the $\ned$ and $\allr$ cases of the inductive step, each case is resolved by invoking IH, and then applying the corresponding rule. In the $\ned$ and $\allr$ cases, one must take care that the freshness condition is not violated; such a situation for $\allr$ is shown below left, where we assume that $z$ is distinct from both $x$ and $t$.
\begin{center}
\begin{tabular}{c @{\hskip 1em} c @{\hskip 1em} c}
\AxiomC{$\dotprf{\prf}$}
\UnaryInfC{$\ns \hol  y, \phi(y/z) \hor$}
\RightLabel{$\allr$}
\UnaryInfC{$\ns \hol  \forall z \phi \hor$}
\RightLabel{$\ps$}
\UnaryInfC{$(\ns \hol \fa z\phi \hor )(t/x)$}
\DisplayProof

&$\leadsto$&

\AxiomC{$\dotprf{\prf}$}
\UnaryInfC{$\ns \hol  y, \phi(y/z) \hor$}
\RightLabel{IH}
\UnaryInfC{$\ns \hol  z', \phi(z'/z) \hor$}
\RightLabel{IH}
\UnaryInfC{$(\ns \hol  z', \phi(z'/z) \hor)(t/x)$}
\RightLabel{$\allr$}
\UnaryInfC{$(\ns \hol  \forall z \phi\hor)(t/x)$}
\DisplayProof
\end{tabular}
\end{center}
To resolve the case above left, we apply IH twice, first replacing $y$ with a fresh variable $z'$, then applying the substitution $(t/x)$, and last, applying $\allr$, as show above right. Observe that rule $\allr$ can be applied since $(z',\phi(z'/z))(t/x)$ is identical with $z',(\phi(t/x))(z'/z)$ thanks to the freshness of $z'$ and the fact that $z$ is distinct from both $t$ and $x$. If, instead, $z$ is $x$ then $(z',\phi(z'/z))(t/x)$ is $z',\phi(z'/z)$. Finally, if $z$ is $t$, then we rewrite $\phi(y/z)$ as $(\phi(z''/z))(y/z'')$ for some fresh $z''$, thus ensuring that $(z',(\phi(z''/z))(z'/z''))(z/x)$ is the same as $z',((\phi(z''/z))(z/x))(z'/z'')$ and, by applying rule $\allr$, we conclude $\ns(z/x)\hol  (\forall z''\phi(z''/z))(z/x)\hor$  which is identical to  $(\ns\{\forall z\phi\hor)(t/x)$ .
\end{proof}

\begin{lemma}[Weakening Rules]\label{lem:wk-admiss}
The $\wk$, $\tw$, $\ew$, and $\nec$ rules are hp-admissible in $\nql$.
\end{lemma}

\begin{proof} First, observe that applying any of the above rules to $\ax$ yields another instance of $\ax$, which establishes the hp-admissibility of each rule in the base case. For the $\ew$ and $\nec$ rules, the inductive step is trivial as both rules permute above all non-initial rules in $\nql$. For the $\wk$ and $\tw$ rules, the inductive step requires slightly more care. With the exception of the $\allr$ and $\ned$ rules, both $\wk$ and $\tw$ permute above all rules of $\nql$. If, however, an application of $\wk$ or $\tw$ is preceded by an application of $\allr$ or $\ned$, then one must ensure that the freshness condition of the latter rules is preserved after permuting $\wk$ or $\tw$ above $\allr$ or $\ned$. This can be accomplished by first applying the hp-admissibility of \lem~\ref{lem:ps-admiss} to substitute a fresh variable that occurs nowhere in the given proof for the fresh variable in the $\allr$ or $\ned$ instance. Afterwards, one can safely permute $\wk$ or $\tw$ above $\allr$ or $\ned$ while preserving the freshness condition.
\end{proof}

\begin{lemma}[Invertibility]\label{lem:hp-inver}
Every rule in $\nql$ is hp-invertible.
\end{lemma}

\begin{proof} The hp-invertibility of $\existsr$, $\diar$, $\nidref$, $\idrep$, $\idrig$, $\dpr$, $\drep$, $\drule$, $\ned$, and $\cdr$ follow from the hp-admissibility of $\wk$, $\tw$, and $\ew$ (\lem~\ref{lem:wk-admiss} above). The other cases ($\disr$, $\conr$, $\allr$, and $\boxr$) are argued as usual by induction on the height of the given proof.
\end{proof}

\begin{lemma}[Contractions]\label{lem:ec-admiss}
The $\tc$, $\ctr$, and $\ec$ rules are hp-admissible in $\nql$.
\end{lemma}

\begin{proof} It is trivial to show that $\tc$ is hp-admissible since any application of $\tc$ to $\ax$ yields another instance of $\ax$, and $\tc$ permutes above every rule in $\nql$. Proving the hp-admissibility of $\ctr$ and $\ec$ requires more work and is proven by simultaneous induction on the height of the given proof. The base cases are simple since any application of $\ctr$ or $\ec$ to $\ax$ gives another instance of $\ax$, and so, we focus on the inductive step. For the inductive step, we show the $\allr$ and $\boxr$ cases and note that the other cases are simpler or similar.

$\allr$. We consider the case where an application of $\ctr$ is preceded by an application of $\allr$ and where the principal formula of $\allr$ is auxiliary in $\ctr$. This case is shown below left and is resolved as shown below right. We first invoke the hp-invertibility of $\allr$ with respect to some fresh  $z$ (\lem~\ref{lem:hp-inver}), then apply the hp-admissibility of $\mathsf{s}(y/z)$ and $\tc$, next we apply IH relative to $\ctr$, and last apply the $\allr$ rule.
\begin{center}
\begin{tabular}{c @{\hskip 1cm} c @{\hskip 1cm} c}
\AxiomC{$\dotprf{\prf}$}
\UnaryInfC{$\ns \hol y, \phi(y/x), \fa x\phi \hor$}
\RightLabel{$\allr$}
\UnaryInfC{$\ns \hol \fa x\phi, \fa x\phi \hor$}
\RightLabel{$\ctr$}
\UnaryInfC{$\ns \hol \fa x\phi \hor$}
\DisplayProof

&$\leadsto$&

\AxiomC{$\dotprf{\prf}$}
\UnaryInfC{$\ns \hol y, \phi(y/x), \fa x\phi \hor$}
\RightLabel{\lem~\ref{lem:hp-inver}} 
\UnaryInfC{$\ns \hol y, z, \phi(y/x), \phi(z/x) \hor$}
\RightLabel{$\mathsf{s}(y/z)$}
\UnaryInfC{$\ns \hol y, y, \phi(y/x), \phi(y/x) \hor$}
\RightLabel{$\tc$}
\UnaryInfC{$\ns \hol y, \phi(y/x), \phi(y/x) \hor$}
\RightLabel{IH}
\UnaryInfC{$\ns \hol y, \phi(y/x) \hor$}
\RightLabel{$\allr$}
\UnaryInfC{$\ns \hol \fa x\phi \hor$}
\DisplayProof
\end{tabular}
\end{center}

$\boxr$. We consider the case where an application of $\ctr$ is preceded by an application of $\boxr$ and where the principal formula of $\boxr$ is auxiliary in $\ctr$. This case is shown below left and is resolved as shown below right. We first invoke the hp-invertibility of $\boxr$ (\lem~\ref{lem:hp-inver}), then apply IH relative to $\ec$, apply IH relative to $\ctr$, and last apply the $\boxr$ rule.
\begin{center}
\begin{tabular}{c @{\hskip 1cm} c @{\hskip 1cm} c}
\AxiomC{$\dotprf{\prf}$}
\UnaryInfC{$\ns \hol \Box \phi, [\phi] \hor$}
\RightLabel{$\boxr$}
\UnaryInfC{$\ns \hol \Box \phi, \Box \phi \hor$}
\RightLabel{$\ctr$}
\UnaryInfC{$\ns \hol \Box \phi \hor$}
\DisplayProof

&$\leadsto$&

\AxiomC{$\dotprf{\prf}$}
\UnaryInfC{$\ns \hol \Box \phi, [\phi] \hor$}
\RightLabel{\lem~\ref{lem:hp-inver}} 
\UnaryInfC{$\ns \hol [\phi], [\phi] \hor$}
\RightLabel{IH}
\UnaryInfC{$\ns \hol [\phi,\phi] \hor$}
\RightLabel{IH}
\UnaryInfC{$\ns \hol [\phi] \hor$}
\RightLabel{$\boxr$}
\UnaryInfC{$\ns \hol \Box \phi \hor$}
\DisplayProof
\end{tabular}
\end{center}
Note that this case demonstrates why $\ctr$ and $\ec$ must be proven hp-admissible \emph{simultaneously} as the hp-admissibility of $\ctr$ in the case depends on that of $\ec$.
\end{proof}

\begin{lemma}[Shift Rule]\label{lem:shift-admiss}
The $\shift$ rule is hp-admissible in $\nql$.
\end{lemma}

\begin{proof} Once again, we proceed by induction on the height of the proof of the premise. If the $\shift$ rule is applied to an initial rule $\ax$  then the result is another instance of $\ax$, showing that the base case of induction goes through. For the inductive step, we consider one instance of the $\diar$ case and one instance of the $\dpr$ case. The remaining cases are analogous or straightforward.

$\diar$. Let us consider an application of $\shift$ preceded by an application of $\diar$, as shown below left. Let $L=L_{\thuesys(\gpset)}(\fd)$. The side condition on the $\diar$ rule ensures that $w \prpath{L} k$ in $\mathcal{G}_0$ and the side condition on rule $\shift$ ensures that $u \prpath{L} v$ in $\mathcal{G}_0$. We have to show that it is possible to permute the instance of $\shift$ above that of $\diar$ without impairing its side condition, thus giving the proof shown below right.
\begin{center}
\begin{tabular}{c  c  c}
\AxiomC{$\dotprf{\prf}$}
\UnaryInfC{$\ns_{0} \hol \dia \phi \hor_{w} \hol  \phi   \hor_{k} \hol \Gamma, [\nsii] \hor_{u} \hol \Delta \hor_{v}$}
\RightLabel{$\diar$}
\UnaryInfC{$\ns_{0} \hol \dia \phi \hor_{w} \hol   \emptyset  \hor_{k} \hol \Gamma, [\nsii] \hor_{u} \hol \Delta \hor_{v}$}
\RightLabel{$\shift$}
\UnaryInfC{$\ns_{1}\hol \dia \phi \hor_{w} \hol  \emptyset  \hor_{k} \hol \Gamma\hor_{u} \hol \Delta, \nsii\hor_{v}$}
\DisplayProof

&$\leadsto$&

\AxiomC{$\dotprf{\prf}$}
\UnaryInfC{$\ns_{0} \hol \dia \phi \hor_{w} \hol  \phi   \hor_{k} \hol \Gamma, [\nsii] \hor_{u} \hol \Delta \hor_{v}$}
\RightLabel{IH}
\UnaryInfC{$\ns_{1}\hol \dia \phi \hor_{w} \hol  \phi  \hor_{k} \hol \Gamma\hor_{u} \hol \Delta, \nsii\hor_{v}$}
\RightLabel{$\diar$}
\UnaryInfC{$\ns_{1}\hol \dia \phi \hor_{w} \hol  \emptyset  \hor_{k} \hol \Gamma\hor_{u} \hol \Delta, \nsii\hor_{v}$}
\DisplayProof
\end{tabular}
\end{center}
If the $w$-component and the $k$-component both occur in $\nsii$ or neither occurs in $\nsii$, then $\shift$ can be freely permuted above $\diar$. Let us suppose then that the $k$-component occurs in $\nsii$ and the $w$-component does not (the other case is similar). Let $u'$ be the name of the component that serves as the root of $\nsii$. It follows that the propagation path from $w$ to $k$ in $\ns_{0}$ is of the form $w \prpath{\stra} u \prpath{\fd} u' \prpath{\strb} k$ with $\stra \fd \strb \in L$. By the side condition of the $\shift$ rule, we know there exists a string $\strc \in L$ such that $u \prpath{\strc} v$. As $L  = L_{\g(\gpset)}(\fd)$, this implies that $\fd \longrightarrow^{*}_{\g(\gpset)} \strc$, which further implies that $\stra \strc \strb \in L$. This string corresponds to the propagation path $w \prpath{\stra} u \prpath{\strc} v \prpath{\strb} k$ in $\ns_{1}$, showing that after $\shift$ has been applied the side condition on $\diar$ still holds, and thus, the two rules may indeed be permuted as shown above right.

$\dpr$. Next, suppose we have an instance of $\dpr$ followed by an instance of $\shift$, as shown below left. We consider a non-trivial case and suppose (1) $\cid \in \fcons$, but $\cdd \not\in \fcons$, and (2) the $k$-component exists in $\nsii$. Note that the side condition on $\dpr$ states that $w \prpath{L} k$ with $L := L_{\sfour \cup \thuesys(\gpset)}(\fd)$. 
\begin{center}
\begin{tabular}{c @{\hskip 1cm} c}
\AxiomC{$\dotprf{\prf}$}
\UnaryInfC{$\ns_{0} \hol t \hor_{w} \hol  t \hor_{k} \hol \Gamma, [\nsii] \hor_{u} \hol \Delta \hor_{v}$}
\RightLabel{$\dpr$}
\UnaryInfC{$\ns_{0} \hol t \hor_{w} \hol  \emptyset \hor_{k} \hol \Gamma, [\nsii] \hor_{u} \hol \Delta \hor_{v}$}
\RightLabel{$\shift$}
\UnaryInfC{$\ns_{1} \hol t \hor_{w} \hol  \emptyset \hor_{k} \hol \Gamma \hor_{u} \hol \Delta, \nsii \hor_{v}$}
\DisplayProof

&

\AxiomC{$\dotprf{\prf}$}
\UnaryInfC{$\ns_{0} \hol t \hor_{w} \hol  t \hor_{k} \hol \Gamma, [\nsii] \hor_{u} \hol \Delta \hor_{v}$}
\RightLabel{$\shift$}
\UnaryInfC{$\ns_{1} \hol t \hor_{w} \hol t \hor_{k} \hol \Gamma \hor_{u} \hol \Delta, \nsii \hor_{v}$}
\RightLabel{$\dpr$}
\UnaryInfC{$\ns_{1} \hol t \hor_{w} \hol  \emptyset \hor_{k} \hol \Gamma \hor_{u} \hol \Delta, \nsii \hor_{v}$}
\DisplayProof
\end{tabular}
\end{center}
By \lem~\ref{lem:S4-with-S(G)-is-S4-or-S5}, we know that $L = L_{\sfour}(\fd)$ or $L = L_{\sfive}(\fd)$. By our assumption and Remark~\ref{rmk:axs-closed-under-consequence}, we need only consider the case when $L = L_{\sfour}(\fd)$. 

If the $w$-component exists in $\nsii$, then the two rules are clearly permutable, so let us suppose that the $w$-component does not occur in $\nsii$. All of this implies that $w \prpath{L} u$ and $u \prpath{L} k$ in $\ns_{0}$. As the shift rule $\shift$ was applied, we know that $\prgr{\ns_{0}} \models u \prpath{L'} v$ with $L' = L_{\g(\gpset)}(\fd)$ by its side condition. Since $L' \subseteq L$, we have that $u \prpath{L} v$ holds in $\ns_{0}$ as well. Let us now permute the two rules and argue that the side condition of $\dpr$ still holds, i.e., that $\prgr{\mathcal{G}_1} \models w \prpath{L} k$.

Since $w \prpath{L} u$ holds in $\ns_{0}$, $w \prpath{L} u$ holds in $\ns_{1}$ as well, and the paths $u \prpath{L} v$ and $v \prpath{L} k$ hold in $\ns_{1}$ by what was said above. (NB. $\prgr{\ns_{1}} \models v \prpath{L} k$ follows from the fact that $\prgr{\ns_{0}} \models u \prpath{L} k$ and $L = L_{\sfour}(\fd)$.) Observe that $\fd \fd \fd \in L$ and that each string associated with a path $w \prpath{L} u$, $u \prpath{L} v$, or $v \prpath{L} k$ is derivable from $\fd$ in $L$. Hence, we have that $w \prpath{L} k$ in $\ns_{1}$, which shows that $\dpr$ may in fact be applied in the proof above right.
\end{proof}

In Remark~\ref{rmk:constant-outer}, we explained that our use of constant outer domains in models of QMLs stemmed from the admissibility of $\ebr$ in each nested calculus. 
We end this subsection by proving this fact. In the proof, one can see that $\ebr$ is essentially an instance of the $\allr$ rule used in our nested systems. This observation 
signifies that if one wishes to construct a nested sequent calculus for a QML characterized by models with \emph{non-constant} outer domains, then one will most likely need to use an alternative rule for the universal quantifier. This suggests that the current nested sequent framework may need to be generalized if one wishes to supply nested sequent calculi for a wider class of QMLs (e.g., various QMLs discussed in Corsi~\cite{Cor02}).

\begin{theorem}[Extended Barcan Rule]\label{theoremEBR}
$\ebr$ is admissible in $\nql$.
\end{theorem}

\begin{proof} To show the admissibility of $\ebr$, we need to show that $\brn{n+1}$ is admissible in $\nql$ for each $n \in \mathbb{N}$. Let $n \in \mathbb{N}$ and suppose that
$$
\phi_{0} \imp \Box (\phi_{1} \imp \cdots \imp \Box (\phi_{n} \imp \Box \phi_{n+1})\ldots)
$$
has a proof $\prf$ such that $x \not\in \fv{\phi_{0}, \ldots, \phi_{n}}$, which gives the top nested sequent in the proof below.
\begin{center}
\AxiomC{$\dotprf{\prf}$}
\UnaryInfC{$\phi_{0} \imp \Box (\phi_{1} \imp \cdots \imp \Box (\phi_{n} \imp \Box\phi_{n+1})\ldots)$}
\RightLabel{=}
\UnaryInfC{$\negnnf{\phi}_{0} \lor \Box (\negnnf{\phi}_{1} \lor \cdots \lor \Box (\negnnf{\phi}_{n} \lor \Box\phi_{n+1})\ldots)$}
\RightLabel{\lem~\ref{lem:hp-inver}}
\UnaryInfC{$\negnnf{\phi}_{0}, [\negnnf{\phi}_{1}, \ldots, [\negnnf{\phi}_{n}, [\phi_{n+1}]]\ldots]$}
\RightLabel{$\tw$}
\UnaryInfC{$\negnnf{\phi}_{0}, [\negnnf{\phi}_{1}, \ldots, [\negnnf{\phi}_{n}, [x,\phi_{n+1}]]\ldots]$}
\RightLabel{$\allr$}
\UnaryInfC{$\negnnf{\phi}_{0}, [\negnnf{\phi}_{1}, \ldots, [\negnnf{\phi}_{n}, [\forall x \phi_{n+1}]]\ldots]$}
\RightLabel{$\boxr$, $\disr$}
\UnaryInfC{$\negnnf{\phi}_{0} \lor \Box (\negnnf{\phi}_{1} \lor \cdots \lor \Box (\negnnf{\phi}_{n} \lor \Box\forall x \phi_{n+1})\ldots)$}
\RightLabel{=}
\UnaryInfC{$\phi_{0} \imp \Box (\phi_{1} \imp \cdots \imp \Box (\phi_{n} \imp \Box\forall x 
\phi_{n+1})\ldots)$}
\DisplayProof
\end{center}
To complete our proof above, we first recall that $\imp$ is defined in terms of negation and disjunction, apply the hp-invertibility of $\boxr$ and $\disr$ (\lem~\ref{lem:hp-inver}), apply $\tw$ followed by the $\allr$ rule, and then apply a sufficient number of $\boxr$ and $\disr$ rule applications to obtain the conclusion of $\brn{n+1}$.
\end{proof}

\subsection{Syntactic Cut-Elimination}\label{subsec:cut-elim}

Despite the seminal works on nested sequent calculi by Kashima~\cite{Kas94} and Bull~\cite{Bul92}, the first direct proofs of syntactic cut-elimination for nested sequent systems (in the context of propositional modal logics) were given by Br\"unnler~\cite{Bru09} and Poggiolesi~\cite{Pog09}. In the first-order setting, the proof of cut-elimination is substantially more complicated because we must keep track of terms, propagation paths, and ensure that side conditions of rules continue to hold after upward permutations of cuts (shown in \thm~\ref{thm:cut-elim} below). Nevertheless, a favorable feature of our cut-elimination proof is that it is \emph{uniform} over the class of logics we consider. That is, our proof does not require the introduction of \emph{ad hoc} structural rules to deal with special cut-elimination cases; cf.~\cite{Bru09,LyoOrl23}. This is primarily due to our use of the hp-admissible shift rule $\shift$, which succinctly unifies reasoning with all generalized path conditions in a single rule.

\begin{theorem}\label{thm:cut-elim}
The following $\cut$ 
rule is admissible in $\nql$.
\smallskip
\begin{center}
\AxiomC{$\ns \hol \phi \hor$}
\AxiomC{$\ns \hol \negnnf{\phi} \hor$}
\RightLabel{$\cut$}
\BinaryInfC{$\ns \hol \emptyset \hor$}
\DisplayProof


\end{center}
\end{theorem}

\begin{proof} 
We prove the result by simultaneous induction on the lexicographic ordering of pairs of the form $(\ell(\phi),h_{1}+h_{2})$, where $\ell(\phi)$ is the length of the cut formula $\phi$, and $h_{1}$ and $h_2$ are the heights of the proofs of the left and right premises of $\cuts$, respectively. We assume w.l.o.g. that $\cut$ 
occurs only once in the given proof and is the last inference; the general result follows by successively applying the cut elimination procedure to topmost instances of $\cut$ 
in a given proof. Note that $\cut$ 
can always be permuted above applications of $\drule$ by utilizing the hp-admissibility of $\ew$, and thus, we may omit consideration of $\drule$ below as all such cases hold. 

We organize the proof into three exhaustive cases: (1) at least one premise of $\cut$ 
is an instance of $\ax$, (2) the cut formula is not principal in at least one premise, and (3) the cut formula is principal in both premises.

(1) Let us suppose that the left premise of $\cut$ is an instance of $\ax$. If the cut-formula is not principal in the left premise, or if the right premise of $\cut$ is an instance of $\ax$ as well, then the conclusion of $\cut$ will be an instance of $\ax$, so the $\cut$ may be deleted entirely. 

Otherwise, the cut-formula is principal in the left premise and we have two subcases to consider: either (i) the cut-formula is not principal in the last inference $\ru$ applied in the proof of the right premise, or (ii) the cut-formula is principal in $\ru$ and $\ru$ is an instance of one of  $\idrep$
, $\drep$ and $\idrig$. In the first  subcase we can permute the cut upward in the proof of the right premise.  We first apply to the left premise either one instance of $\ru^{-1}$ (or two if $\ru$ is an instance of $\conr$), that is hp-admissible by Lemma \ref{lem:hp-inver}. Next, we apply an instance of $\cut$ that is admissible by the inductive hypothesis on the sum of the height of the derivations of the two premises. We conclude by another instance of rule $\ru$. We illustrate this subcase with the following example:

\begin{center}
\AxiomC{$\phantom{\ns}$}
\RightLabel{$\ax$}
\UnaryInfC{$\ns\hol\Box \psi\hor_w\hol L,\negnnf{L}\hor_v$}
\AxiomC{$\dotprf{\prf_{2}}$}
\UnaryInfC{$\ns\hol[\psi]\hor_w\hol \negnnf{L},\negnnf{L}\hor_v$}
\RightLabel{$\boxr$}
\UnaryInfC{$\ns\hol\Box \psi\hor_w\hol \negnnf{L},\negnnf{L}\hor_v$}
\RightLabel{$\cut$}
\BinaryInfC{$\ns\hol[\psi]\hor_w\hol \negnnf{L}\hor_v$}
\DisplayProof
\end{center}
where $w$ and  $v$ need not be  distinct. The case is resolved as shown below:
\begin{center}
\AxiomC{$\phantom{\ns}$}
\RightLabel{$\ax$}
\UnaryInfC{$\ns\hol\Box \psi\hor_w\hol L,\negnnf{L}\hor_v$}
\RightLabel{Lem.~\ref{lem:hp-inver}}
\UnaryInfC{$\ns\hol[\psi]\hor_w\hol {L},\negnnf{L}\hor_v$}
\AxiomC{$\dotprf{\prf_{2}}$}
\UnaryInfC{$\ns\hol[ \psi]\hor_w\hol \negnnf{L},\negnnf{L}\hor_v$}
\RightLabel{IH}
\BinaryInfC{$\ns\hol[\psi]\hor_w\hol \negnnf{L}\hor_v$}
\RightLabel{$\boxr$}
\UnaryInfC{$\ns\hol\Box \psi\hor_w\hol \negnnf{L}\hor_v$}
\DisplayProof
\end{center}

In the second subcase we can obtain a cut-free proof of the same conclusion by applying an hp-admissible instance of $\ctr$ to the conclusion of the right premise of $\cut$,  as shown in the following example, where $w$ and $v$ are distinct:

\begin{flushleft}
\AxiomC{$\phantom{\ns}$}
\RightLabel{$\ax$}
\UnaryInfC{$\ns\hol s\neq t,s=t\hor_w\hol \emptyset\hor_v$}
\AxiomC{$\dotprf{\prf_{2}}$}
\UnaryInfC{$\ns\hol s\neq t,s\neq t\hor_w\hol s\neq t\hor_v$}
\RightLabel{$\idrep$}
\UnaryInfC{$\ns\hol s\neq t,s\neq t\hor_w\hol \emptyset\hor_v$}
\RightLabel{$\cut$}
\BinaryInfC{$\ns\hol s\neq t\hor_w\hol \emptyset\hor_v$}
\DisplayProof
$\quad\leadsto$
\end{flushleft}
\begin{flushright}
\AxiomC{$\dotprf{\prf_{2}}$}
\UnaryInfC{$\ns\hol s\neq t,s\neq t\hor_w\hol s\neq t\hor_v$}
\RightLabel{$\idrep$}
\UnaryInfC{$\ns\hol s\neq t,s\neq t\hor_w\hol \emptyset\hor_v$}
\RightLabel{$\ctr$}
\UnaryInfC{$\ns\hol s\neq t\hor_w\hol \emptyset\hor_v$}
\DisplayProof
\end{flushright}

\medskip

(2) Let us suppose that the cut formula is not principal in the left premise of $\cut$; the case where the cut formula is not principal in the right premise of $\cut$ is argued similarly. Let $\ru$ be the rule proving the left premise of $\cut$. To resolve the case, we apply a $\cut$ between the premise of $\ru$ and the conclusion of the proof obtained by applying the hp-invertibility of $\ru$ to the right premise of $\cut$.  To illustrate, we consider the case where the left premise is by an instance of $\allr$ and $\cut$ is applied, as shown below.

\begin{flushleft}
\AxiomC{$\dotprf{\prf_{1}}$}
\UnaryInfC{$\ns\hol y,\psi(y/x)\hor_w\hol\phi \hor_v$}
\RightLabel{$\allr$}
\UnaryInfC{$\ns\hol\fa x\psi\hor_w\hol\phi\hor_v$}
\AxiomC{$\dotprf{\prf_{2}}$}
\UnaryInfC{$\ns\hol\fa x\psi\hor_w\hol\negnnf{\phi}\hor_v$}
\RightLabel{$\cut$}
\BinaryInfC{$\ns\hol\fa x\psi\hor_w\hol\emptyset\hor_v$}
\DisplayProof
$\quad\leadsto$
\end{flushleft}
\begin{flushright}
\AxiomC{$\dotprf{\prf_{1}}$}
\UnaryInfC{$\ns\hol y,\psi(y/x)\hor_w\hol\phi \hor_v$}
\AxiomC{$\dotprf{\prf_{2}}$}
\UnaryInfC{$\ns\hol\fa x\psi\hor_w\hol\negnnf{\phi}\hor_v$}
\RightLabel{Lem.~\ref{lem:hp-inver}}
\UnaryInfC{$\ns\hol y,\psi(y/x)\hor_w\hol\negnnf{\phi}\hor_v$}
\RightLabel{$\cut$}
\BinaryInfC{$\ns\hol y,\psi(y/x)\hor_w\hol\emptyset \hor_v$}\RightLabel{$\allr$}
\UnaryInfC{$\ns\hol\fa x\psi\hor$}
\DisplayProof
\end{flushright}

(3) We consider only two non-trivial cases and note that the remaining case are simpler or similar. 
First, we consider the case where the cut formula is principal in an application of $\allr$ in the left premise and an application of $\existsr$ in the right premise. 
\begin{center}
\AxiomC{$\dotprf{\prf_{1}}$}
\UnaryInfC{$\ns \hol  t, y, \psi(y/x) \hor$}
\RightLabel{$\allr$}
\UnaryInfC{$\ns \hol  t, \fa x \psi \hor$}

\AxiomC{$\dotprf{\prf_{2}}$}
\UnaryInfC{$\ns \hol  t, \exists x \negnnf{\psi}, \negnnf{\psi(t/x)} \hor$}
\RightLabel{$\existsr$}
\UnaryInfC{$\ns \hol  t, \exists x \negnnf{\psi} \hor$}
\RightLabel{$\cut$}
\BinaryInfC{$\ns \hol t \hor $}
\DisplayProof
\end{center}
To remove the $\cut$, we first reduce the height of the $\cut$ as shown in the left branch of the proof below, relying on the hp-admissible $\wk$ rule. We then use the hp-admissible substitution and term contraction rules as shown in the right branch, and apply a $\cut$ between formulae with a smaller length.
\begin{center}
\AxiomC{$\dotprf{\prf_{1}}$}
\UnaryInfC{$\ns \hol  t, y, \psi(y/x) \hor$}
\RightLabel{$\allr$}
\UnaryInfC{$\ns \hol  t, \fa x \psi \hor$}
\RightLabel{$\wk$}
\UnaryInfC{$\ns \hol  t, \fa x \psi, \negnnf{\psi(t/x)} \hor$}

\AxiomC{$\dotprf{\prf_{2}}$}
\UnaryInfC{$\ns \hol  t, \exists x \negnnf{\psi}, \negnnf{\psi}(t/x) \hor$}
\RightLabel{IH}
\BinaryInfC{$\ns \hol  t, \negnnf{\psi}(t/x) \hor$}

\AxiomC{$\dotprf{\prf_{1}}$}
\UnaryInfC{$\ns \hol  t, y, \psi(y/x) \hor$}
\RightLabel{$\mathsf{s}(t/y)$}
\UnaryInfC{$\ns \hol  t, t, \psi(t/x) \hor$}
\RightLabel{$\tc$}
\UnaryInfC{$\ns \hol  t, \psi(t/x) \hor$}

\RightLabel{IH}
\BinaryInfC{$\ns \hol t \hor $}
\DisplayProof
\end{center}
Last, we show how to resolve the case where the cut formula is principal in an application of $\boxr$ in the left premise and an application of $\diar$ in the right premise. 
\begin{center}
\AxiomC{$\dotprf{\prf_{1}}$}
\UnaryInfC{$\ns \hol [ \psi] \hor_{w} \hol \emptyset \hor_{u}$}
\RightLabel{$\boxr$}
\UnaryInfC{$\ns \hol \Box \psi \hor_{w} \hol \emptyset \hor_{u}$}

\AxiomC{$\dotprf{\prf_{2}}$}
\UnaryInfC{$\ns \hol \Diamond \negnnf{\psi}  \hor_{w} \hol \negnnf{\psi}  \hor_{u}$}
\RightLabel{$\diar$}
\UnaryInfC{$\ns \hol \Diamond \negnnf{\psi}  \hor_{w} \hol \emptyset \hor_{u}$}
\RightLabel{$\cut$}
\BinaryInfC{$\ns \hol \emptyset \hor_{w} \hol \emptyset \hor_{u}$}
\DisplayProof
\end{center}
We first reduce the height of the $\cut$ as shown in the left branch of the proof below, relying on the hp-admissible $\wk$ rule. Due to the side condition on the $\diar$ rule, we know that $w \prpath{L} u$ with $L = L_{\thuesys(\gpset)}(\fd)$. Therefore, we may apply the hp-admissible $\shift$ rule as shown in the right branch of the proof below, and then apply a $\cut$ on the formula $\psi$, which is of a smaller length.
\begin{center}
\AxiomC{$\dotprf{\prf_{1}}$}
\UnaryInfC{$\ns \hol [ \psi] \hor_{w} \hol \emptyset \hor_{u}$}
\RightLabel{$\boxr$}
\UnaryInfC{$\ns \hol \Box \psi \hor_{w} \hol \emptyset \hor_{u}$}
\RightLabel{$\wk$}
\UnaryInfC{$\ns \hol \Box \psi \hor_{w} \hol \negnnf{\psi}\hor_{u}$}

\AxiomC{$\dotprf{\prf_{2}}$}
\UnaryInfC{$\ns\hol \Diamond \negnnf{\psi}  \hor_{w}\hol \negnnf{\psi} \hor_{u}$}
\RightLabel{IH}
\BinaryInfC{$\ns \hol \emptyset \hor_{w} \hol \negnnf{\psi} \hor_{u}$}

\AxiomC{$\dotprf{\prf_{1}}$}
\UnaryInfC{$\ns \hol [ \psi] \hor_{w} \hol \emptyset \hor_{u}$}
\RightLabel{$\shift$}
\UnaryInfC{$\ns \hol \emptyset \hor_{w} \hol \psi \hor_{u}$}

\RightLabel{IH}
\BinaryInfC{$\ns \hol \emptyset \hor_{w} \hol \emptyset \hor_{u}$}

\DisplayProof
\end{center}
This concludes the proof of cut-elimination.
\end{proof}

\subsubsection{Properties of Equality.} We end this section by establishing a few natural properties associated with equality atoms.  Proposition~\ref{prop:sym} shows the admissibility of the \emph{symmetry rule} $\sym$ and of the \emph{transitivity rule} $\tra$. 
Proposition~\ref{prop:id-axs} confirms that nested sequents of a certain form are provable in $\nql$. The significance of these nested sequents is that they serve as analogs of 
standard axiomatic sequents used in sequent calculi for predicate logic with equality (cf.~\cite[Table~6.12]{NegPla11}). Proposition~\ref{prop:admiss-equality-rules} confirms the admissibility of the 
\emph{general replacement rule} $\repgen$, which  corresponds to an intuitive property of equality.

\begin{proposition}\label{prop:sym}
The following rules are admissible in $\nql$:
\begin{center}
\begin{tabular}{@{\hskip 0em} c @{\hskip .25em} c }
\AxiomC{$\ns\hol t\neq s,s\neq t\hor$}
\RightLabel{$\sym$}
\UnaryInfC{$\ns\hol t\neq s\hor$}
\DisplayProof

&

\AxiomC{$\ns\hol t\neq s,s\neq r,t\neq r\hor$}
\RightLabel{$\tra$}
\UnaryInfC{$\ns\hol t\neq s, s\neq r\hor$}
\DisplayProof
\end{tabular}
\end{center}
\end{proposition}
\begin{proof}
The admissibility of $\sym$ and $\tra$ are proven below and rely on the admissibility of the $\wk$ rule (\lem~\ref{lem:wk-admiss}).
\begin{center}
\resizebox{\textwidth}{!}{
\begin{tabular}{c c}
\AxiomC{$\dotprf{\prf}$}
\UnaryInfC{$\ns\hol t\neq s, s\neq t\hor$}
\RightLabel{$\equiv$}
\UnaryInfC{$\ns\hol t \neq s, z \neq t(s/z)\hor$}
\RightLabel{$\wk$}
\UnaryInfC{$\ns\hol t \neq s, z \neq t(s/z), z \neq t(t/z)\hor$}
\RightLabel{$\idrep$}
\UnaryInfC{$\ns\hol t \neq s, z \neq t(t/z)\hor$}
\RightLabel{$\equiv$}
\UnaryInfC{$\ns\hol t \neq s, t \neq t\hor$}
\RightLabel{$\nidref$}
\UnaryInfC{$\ns\hol t \neq s \hor$}
\DisplayProof

&

\AxiomC{$\dotprf{\prf}$}
\UnaryInfC{$\ns\hol t\neq s, s\neq r, t\neq r \hor$}
\RightLabel{$\equiv$}
\UnaryInfC{$\ns\hol t \neq s, z \neq r(s/z), z \neq r(t/z) \hor$}
\RightLabel{$\wk$}
\UnaryInfC{$\ns\hol t \neq s, s \neq t, z \neq r(s/z), z \neq r(t/z) \hor$}
\RightLabel{$\idrep$}
\UnaryInfC{$\ns\hol t \neq s, s \neq t, z \neq r(s/z) \hor$}
\RightLabel{$\sym$}
\UnaryInfC{$\ns\hol t \neq s, z \neq r(s/z) \hor$}
\RightLabel{$\equiv$}
\UnaryInfC{$\ns\hol t \neq s, s \neq r \hor$}
\DisplayProof
\end{tabular}
}
\end{center}
The steps marked with $\equiv$ are syntactical rewriting.
\end{proof}

\begin{proposition}\label{prop:id-axs}
The nested sequents $\ns\hol t=t\hor$ and $\ns\hol t \neq s, \phi(t/z), \negnnf{\phi}(s/z)\hor$ are provable in $\nql$ for any $t \in \ter$ and $\phi \in \lang$.
\end{proposition}

\begin{proof} $\ns\hol t=t\hor$ is provable using $\ax$ and $\nidref$. To show $\ns\hol t \neq s, \phi(t/z), \negnnf{\phi}(s/z)\hor$ is provable in $\nql$, we argue by induction on the length of $\phi$.

\textit{Base case.} The base case is simple to resolve as $\phi$ is a literal and the desired nested sequent can be proven 
as shown below:
\begin{center}
\AxiomC{$\phantom{\ns}$}
\RightLabel{$\ax$}
\UnaryInfC{$\ns \hol s\neq t,t \neq s, P(t/z), \neg{P}(s/z) ,\neg{P}(t/z)\hor$}
\RightLabel{$\idrep$}
\UnaryInfC{$\ns \hol  s\neq t,t \neq s, P(t/z), \neg{P}(s/z)  \hor$}
\RightLabel{$\sym$}
\UnaryInfC{$\ns \hol  t \neq s, P(t/z), \neg{P}(s/z)  \hor$}
\DisplayProof
\end{center}

\textit{Inductive step.} For the inductive step, we consider the cases where $\phi$ is of the form $\forall x \psi$ or $\Box \psi$; the remaining cases are similar. Let $\phi \equiv \forall x \psi$ and note that by IH the top sequent in the proof below is provable. We may assume w.l.o.g. that the fresh variable $y$ used in the $\allr$ is chosen so that $y$ is distinct from $z$, and we may also assume w.l.o.g. that $x \not\in \set{t,s,z}$ since we identify derivations that differ only in the names of bound variables (as discussed on p.~\pageref{note:nested-equiv-renaming}). Based on these assumptions, it follows that $\psi(y/x)(t/z) \equiv \psi(t/z)(y/x)$ and $\negnnf{\psi}(y/x)(s/z) \equiv \negnnf{\psi}(s/z)(y/x)$, which justifies the rewriting step marked with $\equiv$ in the following proof:
\begin{center}
\AxiomC{$\dotprf{\prf}$}
\RightLabel{IH}
\UnaryInfC{$\ns \hol y, t \neq s, \psi(y/x)(t/z), \exists x \negnnf{\psi}(s/z), \negnnf{\psi}(y/x)(s/z) \hor$}
\RightLabel{$\equiv$}
\UnaryInfC{$\ns \hol y, t \neq s, \psi(t/z)(y/x), \exists x \negnnf{\psi}(s/z), \negnnf{\psi}(s/z)(y/x) \hor$}
\RightLabel{$\existsr$}
\UnaryInfC{$\ns \hol y, t \neq s, \psi(t/z)(y/x), \exists x \negnnf{\psi}(s/z) \hor$}
\RightLabel{$\allr$}
\UnaryInfC{$\ns \hol t \neq s, \fa x \psi(t/z), \exists x \negnnf{\psi}(s/z) \hor$}
\DisplayProof
\end{center}
In the case that $\phi\equiv \Box\psi$, the desired conclusion is obtained as shown below.
\begin{center}
\AxiomC{$\dotprf{\prf}$}
\RightLabel{IH}
\UnaryInfC{$\ns \hol  t \neq s ,\dia \negnnf{\psi}(s/z), [t \neq s, \psi(t/z), \negnnf{\psi}(s/z)] \hor$}
\RightLabel{$\idrig$}
\UnaryInfC{$\ns \hol  t \neq s, \dia \negnnf{\psi}(s/z), [\psi(t/z), \negnnf{\psi}(s/z)] \hor$}
\RightLabel{$\diar$}
\UnaryInfC{$\ns \hol  t \neq s, \dia \negnnf{\psi}(s/z), [\psi(t/z)] \hor$}
\RightLabel{$\boxr$}
\UnaryInfC{$\ns \hol  t \neq s, \Box \psi(t/z), \dia \negnnf{\psi}(s/z) \hor$}
\DisplayProof
\end{center}
Note that the top sequent in the proof above is provable by IH.
\end{proof}

\begin{proposition}\label{prop:admiss-equality-rules}
The following rule is admissible in $\nql$:
\begin{center}

\AxiomC{$\ns \hol t \neq s, \phi(t/z), \phi(s/z) \hor$}
\RightLabel{$\repgen$}
\UnaryInfC{$\ns \hol  t \neq s, \phi(t/z)  \hor$}
\DisplayProof

\end{center}
\end{proposition}

\begin{proof} 
The admissibility of $\repgen$ is a consequence of Proposition~\ref{prop:id-axs} and the cut-elimination theorem (\thm~\ref{thm:cut-elim}), as shown below:
\begin{center}
\AxiomC{$\dotprf{\prf}$}
\UnaryInfC{$\ns \hol t \neq s, \phi(t/z), \phi(s/z) \hor$}

\AxiomC{$\phantom{\ns}$}
\RightLabel{Proposition~\ref{prop:id-axs}}
\UnaryInfC{$\ns \hol t \neq s, \phi(t/z), \negnnf{\phi}(s/z) \hor$}
\RightLabel{$\cut$}
\BinaryInfC{$\ns \hol t \neq s, \phi(t/z) \hor$}
\DisplayProof
\end{center}
We remark that 
$\repgen$ can be proven admissible by induction on the height of a given proof, rather than relying on the admissibility of $\cut$. 
\end{proof}

%% file: conclusion2.tex
In this work, we unified Horn-characterizable QMLs supporting equality within a single nested sequent framework. This was achieved by incorporating signatures in nested sequents and using reachability rules in nested systems. The latter rules are unique in that they propagate or consume formulae or terms along paths in a nested sequent, permitting one to toggle between various QMLs by simply changing the parameterizing $\albet$-systems. This allowed us to provide the first sound and cut-free complete nested systems for a broad class of QMLs characterized by relational models that assign inner and outer domains to worlds (cf. Corsi~\cite{Cor02}), and which are subject to a variety of distinct frame and domain conditions. 

In future work, we plan to investigate whether the formula interpretation of nested sequents can be leveraged to extract axiomatic systems for QMLs, in the spirit of~\cite{IshKik07}, which applied a similar approach to sub-intuitionistic predicate logics. We note that the specification of axiomatic systems for various QMLs remains open. 

Another promising line of research involves generalizing our nested sequent framework to capture broader and alternative classes of QMLs. Three directions are particularly noteworthy. First, as discussed throughout the paper, the Extended Barcan Rule ($\ebr$) is entailed by the standard universal quantifier rule, thereby limiting our nested systems to QMLs characterized by models with constant outer domains. It would be interesting to find an alternative $\allr$ rule that allows cut-free nested sequent systems that are sound and complete relative to classes of normal models with non-constant outer domains (cf.~\cite{Cor02}). Second, one could investigate nested systems for QMLs with non-rigid---and possibly non-denoting---terms, drawing inspiration from treatments of definite descriptions based on $\lambda$-abstraction, such as that of Fitting and Mendelsohn~\cite{FitMen98}, which has been explored in the context of labeled calculi~\cite{O21}. Third, it may be fruitful to extend the framework to capture generalizations of the relational semantics considered in this paper, such as counterpart-theoretic semantics (cf.~\cite{BraGhi07,O24}).

%% file: app-proofs.tex
\section{Proofs for \sect~\ref{sec:nested-calculi}}\label{app:proofs-sec-sound-compl}

\begin{customthm}{\ref{thm:soundness-nested}} If $\nql \vdash \ns$, then $\ns$ is valid w.r.t. the class of frames $\fclassc$.
\end{customthm}

\begin{proof} Additional cases for the soundness theorem are provided below:

$\allr$. Suppose that $\nsii = \ns \hol \Gamma, \fa x\phi \hor_{w}$ is invalid. By assumption, there exists a model $\M$ based on a frame in $\fclassc$, an $\M$-assignment $\s$, and an $\M$-interpretation $\mint$ such that $\M, \s, \mint \not\Vdash \nsii$. Hence, $\M, \s, \mint(w) \not\Vdash \fa x \phi$ by \dfn~\ref{def:sequent-semantics}, which implies that there exists an $o \in \D_{\mint(w)}$ such that $\M, \tau, \mint(w) \not\Vdash \phi(y/x)$ with $\tau = \s^{y \trir o}$. In particular, we have that $\tau(y) \in \D_{\mint(w)}$. By \dfn~\ref{def:sequent-semantics} and the assumption that $y$ is fresh, we have that $\M, \tau, \mint \not\Vdash \ns \hol y, \Gamma, \phi(y/x) \hor_{w}$.

$\boxr$. Suppose that $\nsii = \ns \hol  \Gamma, \Box \phi \hor_{w}$ is invalid. By assumption, there exists a model $\M$ based on a frame in $\fclassc$, an $\M$-assignment $\s$, and an $\M$-interpretation $\mint$ such that $\M, \s, \mint \not\Vdash \nsii$, which further implies that $\M, \s, \mint(w) \not\Vdash \Box \phi$. Hence, there exists a world $u \in \W$ of $\M$ such that $\M, \s, u \not\Vdash \phi$. Let $u$ be the name of the component containing just $\phi$ displayed in the premise $\ns \hol \Gamma, [\phi] \hor_{w}$ of $\boxr$. We define $\mint'(v) := \mint(v)$ if $v \neq u$ and let $\mint'(u) := u$ otherwise. Then, one can see that $\M, \s, \mint' \not\Vdash \ns \hol \Gamma, [\phi] \hor_{w}$.

$\nidref$. Suppose that $\nsii = \ns \hol  \Gamma \hor_{w}$ is invalid. By assumption, there exists a model $\M$ based on a frame in $\fclassc$, an $\M$-assignment $\s$, and an $\M$-interpretation $\mint$ such that $\M, \s, \mint \not\Vdash \nsii$. For every term $t$, we have that $\M, \s, \mint(w) \Vdash t  = t$, and thus, $\M, \s, \mint \not\Vdash \ns \hol  \Gamma, t \neq t \hor_{w}$.

$\idrep$. Suppose that $\nsii = \ns \hol  \Gamma, t \neq s, N(t/z)  \hor_{w}$ is invalid. By assumption, there exists a model $\M$ based on a frame in $\fclassc$, an $\M$-assignment $\s$, and an $\M$-interpretation $\mint$ such that $\M, \s, \mint \not\Vdash \nsii$, which means that $\M, \s, \mint(w) \Vdash t = s$ and $\M, \s, \mint(w) \Vdash \negnnf{N}(t/z)$. These two facts imply that $\M, \s, \mint(w) \Vdash \negnnf{N}(s/z)$, which further implies that $\M, \s, \mint \not\Vdash \ns \hol \Gamma, t \neq s, N(t/z),N(s/z) \hor_{w}$.

$\idrig$. Suppose that $\nsii = \ns \hol  \Gamma, s\neq t \hor_{w} \hol \Delta \hor_{u}$ is invalid. By assumption, there exists a model $\M$ based on a frame in $\fclassc$, an $\M$-assignment $\s$, and an $\M$-interpretation $\mint$ such that $\M, \s, \mint \not\Vdash \nsii$, which means that $\M, \s, \mint(w) \Vdash s=t$. By Proposition~\ref{prop:eq-sat-global}, it follows that $\M, \s, \mint(u) \Vdash s=t$, and so, $\M, \s, \mint \not\Vdash \ns \hol  \Gamma, s\neq t \hor_{w} \hol \Delta, s\neq t \hor_{u}$.

$\drep$. Suppose that $\nsii = \ns \hol t, \Gamma, t \neq s \hor_w$ is invalid.  By assumption, there exists a model $\M$ based on a frame in $\fclassc$, an $\M$-assignment $\s$, and an $\M$-interpretation $\mint$ such that $\M, \s, \mint \not\Vdash \nsii$, which means that $\M, \s, \mint(w) \Vdash t = s$ and $\s(t) \in \D_{\mint(w)}$. This implies that $\s(t)=\s(s)$ and, hence, that $\s(s)\in\D_{\mint(w)}$. We conclude that $\M, \s, \mint \not\Vdash  \ns \hol s,t, \Gamma, t \neq s \hor_w$ 

$\ned$. Suppose that $\cned \in \fcons$ and $\nsii = \ns \hol \Gamma \hor_{w}$ is invalid. By assumption, there exists a model $\M$ based on a frame in $\fclassc$, an $\M$-assignment $\s$, and an $\M$-interpretation $\mint$ such that $\M, \s, \mint \not\Vdash \nsii$. By the $\cned$ condition, we know that there exists an object $o \in \D_{\mint(w)}$. Let $y$ be a variable not occurring in $\nsii$, i.e., let $y$ be fresh. Let us define $\tau(t) := \s(t)$ if $t \neq y$ and $\tau(y) := o$ otherwise. Then, $\tau(y) \in \D_{\mint(w)}$ and since $y$ is fresh, we have that $\M, \s, \mint \not\Vdash \ns \hol y, \Gamma \hor_{w}$.

$\cdr$. Suppose that $\ccd \in \fcons$ and $\nsii = \ns \hol \Gamma \hor_{w}$ is invalid. By assumption, there exists a model $\M$ based on a frame in $\fclassc$, an $\M$-assignment $\s$, and an $\M$-interpretation $\mint$ such that $\M, \s, \mint \not\Vdash \nsii$. By the $\ccd$ condition, we know that $\D_{\mint(w)} = \U$. Since $\I_{\mint(w)}(a) \in \U$ for any $a \in \con$ and $\s(x) \in \U$ for any $x \in \var$, we know that $\I_{\mint(w)}^{\s}(t) \in \U = \D_{\mint(w)}$ for any $t \in \ter$. Therefore, $\M, \s, \mint \not\Vdash \ns \hol t, \Gamma \hor_{w}$.
\end{proof}

\begin{customthm}{\ref{thm:completeness-nested}} If $\ns$ is valid w.r.t. the class of frames $\fclassc$, then $\nql \vdash \ns$.
\end{customthm}

\begin{proof} We suppose that $\nql$ does not prove a nested sequent $\ns$ and show that $\ns$ is invalid relative to $\fclassc$. To prove this, we first define a proof-search procedure $\prove$ that bottom-up applies rules from $\nql$ to $\ns$ building an infinite derivation thereof. Second, we show how a model $\M = \langle \W,\R,\U,\D,\I \rangle$ can be extracted from this infinite derivation. Let $w$ be the name of the root of $\ns$ and let $\prec$ be a well-founded, strict linear order over the set $\ter$. We now describe the proof-search procedure $\prove$.\\

\noindent
$\prove$. Let us take $\ns$ as input and continue to the next step.\\

$\ax$. Let $\branch_{1}, \ldots, \branch_{n}$ be all branches occurring in the current derivation and let $\ns_{1}, \ldots, \ns_{n}$ be the top sequents of each respective branch. For each $1 \leq i \leq n$, we halt the computation of $\prove$ on each branch $\branch_{i}$ where $\ns_{i}$ is an initial sequent. If $\prove$ halts on each branch $\branch_{i}$, then $\prove$ returns $\true$ because a proof of the input has been constructed. However, if $\prove$ does not halt on each branch $\branch_{i}$ with $1 \leq i \leq n$, then let $\branch_{j_{1}}, \ldots, \branch_{j_{k}}$ be the remaining branches for which $\prove$ does not halt. For each such branch, copy the top sequent above itself, and continue to the next step.

\medskip

$\diar$. Let $\branch_{1}, \ldots, \branch_{n}$ be all branches occurring in the current derivation and let $\ns_{1}, \ldots, \ns_{n}$ be the top sequents of each respective branch. For each $1 \leq i \leq n$, we consider $\branch_{i}$ and extend the branch with bottom-up applications of $\diar$ rules. Let $\branch_{k+1}$ be the current branch under consideration, and assume that $\branch_{1},\ldots,\branch_{k}$ have already been considered. Let $\dia \phi_{1}, \ldots, \dia \phi_{m}$ be all diamond formulae occurring in $\ns_{k+1}$ occurring in components named $w_{1}, \ldots, w_{m}$, respectively. We consider each formula $\Diamond \phi_{i}$ in turn and bottom-up apply the $\diar$ rule for each name $u$ such that $w_{i} \prpath{L} u$ with $L = L_{\thuesys(\gpset)}(\fd)$. After these operations have been performed for each branch $\branch_{i}$ with $1 \leq i \leq n$, we continue to the next step.

\medskip

$\allr$. Let $\branch_{1}, \ldots, \branch_{n}$ be all branches occurring in the current derivation and let $\ns_{1}, \ldots, \ns_{n}$ be the top sequents of each respective branch. For each $1 \leq i \leq n$, we consider $\branch_{i}$ and extend the branch with bottom-up applications of $\allr$ rules. Let $\branch_{k+1}$ be the current branch under consideration, and assume that $\branch_{1},\ldots,\branch_{k}$ have already been considered. Let $\forall x_{1} \phi_{1}, \ldots, \forall x_{m} \phi_{m}$ be all universal formulae occurring in $\ns_{k+1}$. We consider each formula $\forall x_{i} \phi_{i}$ in turn, and bottom-up apply the $\allr$ rule with a fresh variable $z$ that does not occur anywhere in the derivation under construction. After these operations have been performed for each branch $\branch_{i}$ with $1 \leq i \leq n$, we continue to the next step.

\medskip

$\nidref$. Let $\branch_{1}, \ldots, \branch_{n}$ be all branches occurring in the current derivation and let $\ns_{1}, \ldots, \ns_{n}$ be the top sequents of each respective branch. For each $1 \leq i \leq n$, we consider $\branch_{i}$ and extend the branch with bottom-up applications of $\nidref$ rules. Let $\branch_{k+1}$ be the current branch under consideration, and assume that $\branch_{1},\ldots,\branch_{k}$ have already been considered. Let $w_{1}, \ldots, w_{m}$ be the names of all components in $\ns_{k+1}$. For each $w_{j}$-component (with $1 \leq j \leq m$), we pick the smallest term $t$ according to the ordering $\prec$ such that $t \neq t$ does not occur in the $w_{j}$-component and we bottom-up apply the $\nidref$ rule to introduce $t \neq t$ to the $w_{j}$-component. After these operations have been performed for each branch $\branch_{i}$ with $1 \leq i \leq n$, we continue to the next step.

\medskip

$\idrep$. Let $\branch_{1}, \ldots, \branch_{n}$ be all branches occurring in the current derivation and let $\ns_{1}, \ldots, \ns_{n}$ be the top sequents of each respective branch. For each $1 \leq i \leq n$, we consider $\branch_{i}$ and extend the branch with bottom-up applications of $\idrep$ rules. Let $\branch_{k+1}$ be the current branch under consideration, and assume that $\branch_{1},\ldots,\branch_{k}$ have already been considered. Let $w_{1}, \ldots, w_{m}$ be the names of all components in $\ns_{k+1}$. For each $w_{j}$-component (with $1 \leq j \leq m$) and each pair of formulae $t \neq s$ and $N(t/z)$, we bottom-up apply the $\idrep$ rule to introduce $N(s/z)$ to the $w_{j}$-component. After these operations have been performed for each branch $\branch_{i}$ with $1 \leq i \leq n$, we continue to the next step.

\medskip

$\idrig$. Let $\branch_{1}, \ldots, \branch_{n}$ be all branches occurring in the current derivation and let $\ns_{1}, \ldots, \ns_{n}$ be the top sequents of each respective branch. For each $1 \leq i \leq n$, we consider $\branch_{i}$ and extend the branch with bottom-up applications of $\idrig$ rules. Let $\branch_{k+1}$ be the current branch under consideration, and assume that $\branch_{1},\ldots,\branch_{k}$ have already been considered. Let $w_{1}, \ldots, w_{m}$ be the names of all components in $\ns_{k+1}$. For each $w_{j}$-component (with $1 \leq j \leq m$) and each equality literal $t \neq s$ occurring in the $w_{j}$-component, we bottom-up apply the $\idrig$ rule $m-1$ times to add $t \neq s$ to each $w_{\ell}$-component with $1 \leq \ell \neq j \leq m$. After these operations have been performed for each branch $\branch_{i}$ with $1 \leq i \leq n$, we continue to the next step.



\medskip

$\dpr$. If $\dpr \in \nql$, then we do the following: Let $\branch_{1}, \ldots, \branch_{n}$ be all branches occurring in the current derivation and let $\ns_{1}, \ldots, \ns_{n}$ be the top sequents of each respective branch. For each $1 \leq i \leq n$, we consider $\branch_{i}$ and extend the branch with bottom-up applications of $\dpr$ rules. Let $\branch_{k+1}$ be the current branch under consideration, and assume that $\branch_{1},\ldots,\branch_{k}$ have already been considered. Let $w_{1}, \ldots, w_{m}$ be the names of all components of $\ns_{k+1}$. We consider each $w_{j}$-component in turn and bottom-up apply the $\dpr$ rule for each term $t$ occurring in the $w_{j}$-component and each name $u$ such that $w_{j} \prpath{L} u$ with $L$ determined by the side condition of $\dpr$. After these operations have been performed for each branch $\branch_{i}$ with $1 \leq i \leq n$, we continue to the next step.



\medskip

$\ned$. If $\ned \in \nql$, then we do the following: If $\dpr \in \nql$, then we do the following: Let $\branch_{1}, \ldots, \branch_{n}$ be all branches occurring in the current derivation and let $\ns_{1}, \ldots, \ns_{n}$ be the top sequents of each respective branch. For each $1 \leq i \leq n$, we consider $\branch_{i}$ and extend the branch with bottom-up applications of $\dpr$ rules. Let $\branch_{k+1}$ be the current branch under consideration, and assume that $\branch_{1},\ldots,\branch_{k}$ have already been considered. Let $w_{1}, \ldots, w_{m}$ be the names of all components of $\ns_{k+1}$. We consider each $w_{j}$-component in turn and bottom-up apply the $\ned$ rule introducing fresh variable $y$ that does not occur anywhere in the derivation under construction. After these operations have been performed for each branch $\branch_{i}$ with $1 \leq i \leq n$, we continue to the next step.

\medskip

$\cdr$. If $\cdr \in \nql$, then we do the following: If $\ned \in \nql$, then we do the following: If $\dpr \in \nql$, then we do the following: Let $\branch_{1}, \ldots, \branch_{n}$ be all branches occurring in the current derivation and let $\ns_{1}, \ldots, \ns_{n}$ be the top sequents of each respective branch. For each $1 \leq i \leq n$, we consider $\branch_{i}$ and extend the branch with bottom-up applications of $\cdr$ rules. Let $\branch_{k+1}$ be the current branch under consideration, and assume that $\branch_{1},\ldots,\branch_{k}$ have already been considered. Let $w_{1}, \ldots, w_{m}$ be the names of all components in $\ns_{k+1}$. For each $w_{j}$-component (with $1 \leq j \leq m$), we pick the smallest term $t$ according to the ordering $\prec$ such that $t$ does not occur in the $w_{j}$-component and we bottom-up apply the $\cdr$ rule to introduce $t$ to the $w_{j}$-component. After these operations have been performed for each branch $\branch_{i}$ with $1 \leq i \leq n$, we continue to the next step.

\medskip

$\lor$, $\land$, $\boxr$, $\existsr$, $\drep$, $\drule$. The remaining cases are handled similarly. After these operations have been performed for each branch $\branch_{i}$ with $1 \leq i \leq n$, we loop back to the `$\ax$' step.

\medskip

\noindent
This concludes the description of $\prove$.\\

We now construct a model $\M = \langle \W,\,\R,\,\U,\,\D,\I \rangle$, $\M$-assignment $\s$, and $\M$-interpretation $\mint$ such that $\M, \s, \mint \not\Vdash \ns$. By assumption, $\ns$ does not have a proof in $\nql$, and so, $\prove$ cannot halt/return $\true$. Otherwise, a proof for $\ns$ can be obtained by `contracting'  all redundant inferences from the $\ax$ step of $\prove$, contradicting our assumption. Therefore, as $\prove$ does not halt, $\prove$ generates an infinite tree with finite branching, i.e., an infinite derivation. By K\"onig's lemma, an infinite branch $\branch = \ns_{0}, \ns_{1}, \ldots, \ns_{n}, \ldots$ must exist in this infinite derivation with $\ns_{0} = \ns$. We now define a model $\M = \langle \W,\R,\U,\D,\I \rangle$ by means of this branch: First, for each nested sequent in the branch $\branch$, we let $\tr(\ns_{i}) = (V_{i},E_{i})$ and $\prgr{\ns_{i}} = (\prv_{i}, \pre_{i}, \prl_{i})$ (see Definitions \ref{definition:tns} and \ref{def:propagation-graph}). We define $\tr(\branch) = (V^{\branch},E^{\branch})$ such that 

\begin{itemize}

\item[(1)] $(u,\Gamma) \in V^{\branch}$ with $\Gamma := \bigcup_{i \in \mathbb{N}} \Gamma_{i}$ for $(u,\Gamma_{i}) \in V_{i}$;

\item[(2)] $E^{\branch} = \bigcup_{i \in \mathbb{N}} E_{i}$. 

\end{itemize}
We define $\prgr{\branch} = (\prv^{\branch}, \pre^{\branch}, \prl^{\branch})$ such that
\begin{itemize}

\item[(1)] $\prv^{\branch} = \bigcup_{i \in \mathbb{N}} \prv_{i}$, i.e., $\prv^\branch=\{u: (u,\Gamma)\in V^\branch$ for some $\Gamma \}$; 

\item[(2)] $\pre^{\branch} = \bigcup_{i \in \mathbb{N}} \pre_{i}$, i.e., $\pre^\branch=\{(w, \fd, u), (u, \bd, w) :(w,u) \in E^\branch\}$; 

\item[(3)] for each $u \in \prv^{\branch}$, $\prl^{\branch}(u) = \bigcup_{i \in \mathbb{N}} \prl_{i}(u)$, where we take $\prl_{i}(u) := \emptyset$ if $u \not\in \prv_{i}$.

\end{itemize}

The universe $\U$ and each domain associated with a world of the model $\M$ will consist of equivalence classes of terms. To define these equivalence classes we first define the following equivalence relation: $t \sim s$ \iffi there exists a $(u,\Gamma) \in V^{\branch}$ with $t \neq s \in \Gamma$. We define $[t] = \{s \mid t \sim s\}$. Let us prove that $\sim$ is indeed an equivalence relation. First, by the $\nidref$ step of $\prove$, we know that for any $t \in \ter$, $t \neq t \in \Gamma$ with $(u,\Gamma) \in V^{\branch}$, showing that $t \sim t$, i.e., $\sim$ is reflexive. Second, we show that $\sim$ is Euclidean, i.e., if $t \sim s$ and $t \sim r$, then $s \sim r$. Suppose that $t \sim s$ and $t \sim r$, meaning, $t \neq s \in \Gamma$ and $t \neq r \in \Delta$ for $(u,\Gamma), (v,\Delta) \in V^{\branch}$. By the $\idrig$ step of $\prove$, we have that $t \neq r \in \Gamma$ as well; therefore, for some $i \in \mathbb{N}$, $\ns_{i} = \ns \hol  \Sigma, t \neq s, t \neq r  \hor_{u}$ and by the $\idrep$ step of $\prove$, the $\idrep$ rule will be applied as shown below with $N = (z \neq r)$ thus introducing (bottom-up) $s \neq r$. This implies that $s \neq r \in \Gamma$, showing that $s \sim r$, and establishing that $\sim$ is an equivalence relation.
\begin{center}
\AxiomC{$\ns \hol \Gamma, t \neq s, t \neq r, s \neq r \hor_{u}$}
\RightLabel{$\equiv$}
\UnaryInfC{$\ns \hol \Gamma, t \neq s, (z \neq r)(t/z), (z \neq r)(s/z) \hor_{u}$}
\RightLabel{$\idrep$}
\UnaryInfC{$\ns \hol  \Gamma, t \neq s, (z \neq r)(t/z)  \hor_{u}$}
\RightLabel{$\equiv$}
\UnaryInfC{$\ns \hol  \Gamma, t \neq s, t \neq r  \hor_{u}$}
\DisplayProof
\end{center}
We now define the model $\M = \langle \W,\,\R,\,\U,\,\D,\I \rangle$ as follows:
\begin{itemize}

\item $\W := \prv^{\branch}$;

\item $(w,u) \in \R$ \iffi $\prgr{\branch} \models w \prpath{L} u$ with $L = L_{\thuesys(\gpset)}(\fd)$;

\item $\U :=  \ter \ / \sim$;

\item for each $u \in V^{\branch}$, $\D_{u} := \prl^{\branch}(u) \ / \sim$;

\item $([t_{1}], \ldots, [t_{n}]) \in \I_{u}(P)$ \iffi $(u,\Gamma) \in V^{\branch}$ with $\neg P(t_{1}, \ldots, t_{n}) \in \Gamma$.

\item $\I_{u}(a) := [a]$.

\end{itemize}
Let us first verify that $\M$ is indeed a model (see \dfn~\ref{def:model}). After, we prove that $\M$ is based on a frame in $\fclassc$.

\begin{itemize}

\item Since $w$ is the name of the root of $\ns$, we know that $w \in \W$, and therefore, $\W \neq \emptyset$.

\item We have that $\R \subseteq \W \times \W$ by definition.

\item By definition, the universe $\U$ will be non-empty.


\item We now show that for each $u \in \W$, $\D_{u} \subseteq \U$. Therefore, we need to show that for any term $t$, if $[t]_{\D} := [t] \in \D_{u}$, then $[t]_{\U} := [t] \in \U$ and $[t]_{\D} = [t]_{\U}$. We use $[t]_{\D}$ to denote the equivalence class $[t] \in \D_{u}$ and $[t]_{\U}$ to denote $[t] \in \U$, which disambiguates the two occurrences of $[t]$. Suppose that $[t]_{\D} = [t] \in \D_{u}$. By definition, $[t]_{\U} := [t] \in \U$, though the question remains if $[t]_{\D} = [t]_{\U}$, which we now show. It is trivial that $[t]_{\D} \subseteq [t]_{\U}$, so we prove that $[t]_{\U} \subseteq [t]_{\D}$. Let $s \in [t]_{\U}$, meaning, there exists a $(v,\Delta) \in V^{\branch}$ with $t \neq s \in \Delta$. Note that, by our supposition, $t \in \Gamma$ for some $(u,\Gamma) \in V^{\branch}$. By the $\idrig$ step of $\prove$, we know that $t \neq s \in \Gamma$ for $(u,\Gamma) \in V^{\branch}$. Hence, $t \in \Gamma$ and $t \neq s \in \Gamma$, meaning, $s \in \Gamma$ by the $\drep$ step of $\prove$. It follows that $s \in [t]_{\D}$, which shows that $[t]_{\U} \subseteq [t]_{\D}$.

\item By the definition of $\I$ it is straightforward to verify that for each $a \in \con$, if $u \R v$, then $\I_{u}(a) = \I_{v}(a)$. This follows from that fact that $\I_{u}(a) = [a] = \I_{v}(a)$ holds for each $a \in \con$.

\item Last, $\I_{u}$ is well-defined, that is, for $t_{i}, s_{i} \in [t_{i}]$, $([t_{1}], \ldots, [t_{n}]) \in \I_{u}(P)$ \iffi $([s_{1}], \ldots, [s_{n}]) \in \I_{u}(P)$. This follows from the $\idrig$ and $\idrep$ steps of $\prove$.

\end{itemize}

The above arguments establish that $\M$ is indeed a model. We now show that $\M$ is based on a frame in $\fclassc$. We show that if a frame condition occurs in $\fcset$, then $\M$ satisfies that frame condition.

\begin{description}

\item[$\cser$] If $\cser \in \fcset$, then $\drule \in \nql$, which ensures that for every $u \in \W$, there exists a $v \in \W$ such that $u \R v$.

\item[$\gpc(n,k)$] Let $\gpc(n,k) \in \fcset$ and assume that $w \R^{n} u$ and $w \R^{k} v$. We aim to show that $u \R v$. Since $\gpc(n,k) \in \fcset$, we know that $(\fd \pto \bd^{n} \cate \fd^{k}), (\bd \pto \bd^{k} \cate \fd^{n}) \in \g(\gpset)$ with $\gpset$ the set of all generalized path conditions in $\fcset$. Since $w \R^{n} u$ and $w \R^{k} v$, $\prgr{\branch} \models u \prpath{\bd^{n}} w$ and $\prgr{\branch} \models w \prpath{\fd^{k}} v$ by the definition of $\R$, implying $\prgr{\branch} \models u \ \prpath{\bd^{n} \fd^{k}} \ v$. As $\bd^{n} \fd^{k} \in L_{\thuesys(\gpset)}(\fd)$, it follows that $(u,v) \in \R$ by definition, i.e., $u \R v$.

\item[$\cid$] If $\cid \in \fcons$, then $\dpr \in \nql$. Assume that $u \R v$ and $[t] \in \D_{u}$. We aim to show that $[t] \in \D_{v}$. By assumption, $t \in \Gamma$ with $(u,\Gamma) \in V^{\branch}$. Since $u \R v$, we know that $\prgr{\branch} \models u \prpath{\fd} v$. Regardless of if $L = L_{\sfour \cup \thuesys(\gpset)}(\fd)$ or $L = L_{\sfive(\gpset)}(\fd)$, it follows that $u \prpath{L} v$, and so, $\dpr$ will be (bottom-up) applied at some point in $\prove$, meaning $t \in \Delta$ for $(v,\Delta) \in V^{\branch}$. Let us use $[t]_{u}$ to denote $[t] \in \D_{u}$ and $[t]_{v}$ to denote $[t] \in \D_{v}$. We argue next that $[t]_{u} = [t]_{v}$. We first argue that $[t]_{u} \subseteq [t]_{v}$. Let $s \in [t]_{u}$. Then, $t \neq s \in \Gamma$, and so, $t \neq s \in \Delta$ by $\idrig$, which shows that $s \in [t]_{v}$. This proves that $[t]_{u} \subseteq [t]_{v}$ and we note that $[t]_{v} \subseteq [t]_{u}$ can be established by a similar argument. Therefore, since $t \in \Delta$ and $[t]_{u} = [t]_{v}$, we have that $[t] \in \D_{v}$, which proves that $\cid$ holds.

\item[$\cdd$] Similar to the argument for $\cid$ above.

\item[$\ccd$] Let $\ccd \in \fcons$. Then, $\cdr \in \nql$. By the $\cdr$ step of $\prove$ we know that $\prl^{\branch}(u) = \ter$ for every $u \in \prv^{\branch} = \W$. Therefore, for every $u \in \W$, $\D_{u} = \ter \slash \sim \ = \U$, which shows that $\D_{u} = \U$ for every $u \in \W$.  

\item[$\cned$] Let $\cned \in \fcons$. Then, $\ned \in \nql$. By the $\ned$ step of $\prove$ we know that $\prl^{\branch}(u) \neq \emptyset$ for every $u \in \prv^{\branch} = \W$. Therefore, for every $u \in \W$, $\D_{u} \neq \emptyset$.

\end{description}

We have now confirmed that $\M$ is indeed a model based on a frame in $\fclassc$. Let us define $\s$ to be the $\M$-assignment mapping every variable $x$ to $[x] \in \U$. To finish the proof of completeness, we now argue the following by induction on the lenght of $\phi$: for every $u \in \W$, if $\phi \in \Gamma$ with $(u,\Gamma) \in V^{\branch}$, then $\M,\s,u \not\Vdash \phi$. We show a representative number of cases.
 
\begin{itemize}

\item Let $\phi \equiv P(t_{1},\ldots,t_{n}) \in \Gamma$. Then, $\neg P(t_{1},\ldots,t_{n}) \not\in \Gamma$ since otherwise $\ax$ would be applied along the branch $\branch$, meaning, $\branch$ would not be infinite contrary to our assumption. Therefore, $([t_{1}],\ldots,[t_{n}]) \not\in \I_{u}^{\s}(P)$ by the definition of $\I$, implying that $\M , \s, u \not\Vdash P(t_{1},\ldots,t_{n})$.

\item Let $\phi \equiv \neg P(t_{1},\ldots,t_{n}) \in \Gamma$. Then, $([t_{1}],\ldots,[t_{n}]) \in \V(P,u)$ by the definition of $\V$, implying that $\M , \s, u \not\Vdash \neg P(t_{1},\ldots,t_{n})$.

\item Let $\phi \equiv (t = s) \in \Gamma$. Then, $t \neq s \not\in \Gamma$ since otherwise $\ax$ would be applied along the branch $\branch$, meaning, $\branch$ would not be infinite contrary to our assumption. As a consequence, $s \not\in [t]$. Therefore, $\I_{u}^{\s}(t) = [t] \neq [s] = \I_{u}^{\s}(s)$, implying that $\M , \s, u \not\Vdash t = s$.

\item Let $\phi \equiv (t \neq s) \in \Gamma$. Then, $s \in [t]$, i.e. $[t] = [s]$ by definition, meaning, $\I_{u}^{\s}(t) = [t] = [s] = \I_{u}^{\s}(s)$, implying that $\M , \s, u \not\Vdash t \neq s$.

\item Let $\phi \equiv \exists x \psi \in \Gamma$. 
Let $o \in \D_{u}$. Then, we know that there exists a term $t$ such that $[t] = o$, meaning, $t \in \prl(u)$. By our assumption that $\exists x \psi \in \Gamma$, by the definition of $\prgr{\branch}$, and by the $\existsr$ step of $\prove$, we have that some $j$ exists such that both $\exists x \psi$ and $t$ occur in the $u$-component of $\ns_{j}$ with $\psi(t/x)$ being introduced by a bottom-up application of $\existsr$ to the $u$-component of $\ns_{j+1}$. It follows that $\psi(t/x) \in \Gamma$, so by IH, we have that $\M, \s^{t \trir o}, u \not\models \psi(t/x)$ for any $o \in \D_{u}$ as $o$ was taken to be arbitrary. Hence, $\M, \s, u \not\models \exists x \psi$.

\item Let $\phi \equiv \forall x \psi \in \Gamma$. Then, by the $\allr$ step of $\prove$, we know that for some variable $y$, we have $y, \psi(y/x) \in \Gamma$. It follows that $[y] \in \D_{u}$, and so, $\M, \s, u \not\Vdash \psi(y/x)$ by IH. Hence, $\M, \s, u \not\models \forall x \psi$.

\item Let $\phi \equiv \dia \psi \in \Gamma$. Let $v$ be an arbitrary world in $\W$ such that $u \R v$. Then, $\prgr{\branch} \models u \prpath{L} v$ with $L = L_{\thuesys(\gpset)}(\fd)$ by definition. By the definition of $\prove$, the $\diar$ rule will be bottom-up applied at some point along the branch $\branch$, thus ensuring that $\psi \in \Delta$ for $(v, \Delta) \in V^{\branch}$. By IH, we have that $\M, \s, v \not\Vdash \psi$. As $v$ was arbitrary, this implies $\M, \s, u \not\Vdash \dia \psi$.

\item Let $\phi \equiv \Box \psi \in \Gamma$. Then, by the $\boxr$ step of $\prove$, we know that some nesting with a name $v$ will be introduced containing $\psi$, i.e., $\psi \in \Delta$ for $(v,\Delta) \in V^{\branch}$. By IH, we have that $\M, \s, v \not\Vdash \psi$, which shows that $\M, \s, u \not\Vdash \Box \psi$.

\end{itemize}

Let $\mint$ to be the $\M$-interpretation such that $\iota(u) = u$ for $u \in \W$ and $\mint(v) \in \W$ for $v \in \names \setminus \W$. By the proof above, $\M, \s , \mint \not\Vdash \ns$, showing that if a nested sequent is not provable in $\nql$, then it is invalid, that is, every valid nested sequent is provable in $\nql$.
\end{proof}